\documentclass[11pt]{revtex4-1}

\usepackage{graphicx}
\usepackage{dcolumn}
\usepackage{bm}
\usepackage{amsmath}
\usepackage{amsthm}
\usepackage{amsfonts}
\usepackage{amssymb}
\usepackage{comment}
\usepackage[T1]{fontenc}

\newtheorem{Definition}{Definition}
\newtheorem{Theorem}{Theorem}
\newtheorem{Lemma}{Lemma}
\newtheorem{Proposition}{Proposition}
\newtheorem{Corollary}{Corollary}

\DeclareMathOperator{\Span}{span}

\def \equals {\ = \ }
\def \plus {\ + \ }
\def \Borland {\wedge^3\mathbb{C}^6}
\def \Csix {\mathbb{C}^6}

\def \HH {\mathcal{H}}
\def \LL {\mathcal{L}}
\def \MM {\mathcal{M}}

\def \QQ {\mathcal{Q}}
\def \RR {\mathcal{R}}

\def \VV {\mathcal{V}}
\def \WW {\mathcal{W}}

\def \Ffrak {\mathfrak{F}}

\begin{document}

\title{Concise configuration interaction expansions for three fermions in six orbitals}

\author{Alex D. Gottlieb}
 \email{alex@alexgottlieb.com}
\author{Norbert J. Mauser}%
\affiliation{Wolfgang Pauli Institute c/o Fakult\"at f\"ur Mathematik, Universit\"at Wien, Oskar-Morgenstern Platz 1, 1090 Vienna, Austria}%

\author{J.~M.~Zhang}
 \email{wdlang06@gmail.com}
\affiliation{College of Physics and Energy, Fujian Normal University, Fujian Provincial Key Laboratory of Quantum Manipulation and New Energy Materials, Fuzhou, 350117, China}
\affiliation{Fujian Provincial Collaborative Innovation Center for Optoelectronic Semiconductors and Efficient Devices, Xiamen, 361005, China}

\date{\today}

\begin{abstract}
The Hilbert space for three fermions in six orbitals, lately dubbed the ``Borland-Dennis setting,'' is a proving ground for insights into electronic structure.   
Borland and Dennis discovered that, when referred to coordinate systems defined in terms of its natural orbitals, a wave function in the Borland-Dennis setting has the same structure as a $3$-qubit state.  By dint of the Borland-Dennis Theorem, canonical forms for $3$-qubit states have analogs in the Borland-Dennis setting.   

One of these canonical forms is based upon ``max-overlap Slater determinant approximations.'' 
Any max-overlap Slater determinant approximation of a given wave function is the leading term in a $5$-term configuration interaction (CI) expansion of that wave function.  
Our main result is that ``max-overlap CIS approximations'' also lead to $5$-term CI expansions, distinct from those based on max-overlap Slater determinant approximations, though of the same symmetric shape.  
We also prove the analog of this result for $3$-qubit setting.  
\end{abstract}

\pacs{Valid PACS appear here}
\maketitle

\section{Introduction}

The Hilbert space for three fermions in six orbitals has lately been dubbed the ``Borland-Dennis setting'' \cite{SchillingGrossChristandl,Schilling}.  
The founding contribution of Borland and Dennis was their discovery that every wave function in this Hilbert space 
has a tidy configuration interaction (CI) expansion with respect to its  natural orbitals, and that this entails constraints on its natural occupation numbers \cite{borland1,borland2}.   
The Borland-Dennis setting is a proving ground for insights into electronic structure because it is the smallest ``nontrivial'' many-fermion Hilbert space in some sense.  

As a vector space, the Borland-Dennis space is isomorphic to $\Borland$, the space of alternating $3$-tensors on $\mathbb{C}^6$.   Hilbert space geometry is imposed on $\Borland$ by declaring the decomposable tensors $e_i  \wedge  e_j  \wedge  e_k$ with $i < j < k$ to be orthonormal, where $\{e_1,\ldots,e_6\}$ is the standard basis of $\mathbb{C}^6$.     

Every vector in $\Borland$ can be written as a linear combination of the twenty basis vectors $e_i  \wedge  e_j  \wedge  e_k$; this is its CI expansion with respect to the orthonormal basis $\{e_1,\ldots,e_6\}$ of ``reference'' orbitals.   A different orthonormal basis of $\Csix$ would yield a different CI expansion; changes-of-basis of $\Csix$ induce changes-of-basis of $\Borland$.  Thus the group U(6) of $6 \times 6$ unitary matrices acts on the $20$-dimensional vector space of CI expansion coefficients.  Recently, Chen, \DJ okovi\'c, Grassl, and Zeng have presented a complete set of algebraic invariants for this action of U(6) on the Borland-Dennis space \cite{ChenDokovicGrasslZeng}.

The Borland-Dennis setting contains an isometric image of the $3$-qubit Hilbert space 
$\mathbb{C}^2 \otimes \mathbb{C}^2 \otimes \mathbb{C}^2$ via the embeddding 
\begin{eqnarray}
  |000\rangle \  \longmapsto \   e_1 \wedge  e_3 \wedge  e_5 , 
  & \quad &
  |111\rangle \  \longmapsto \   e_2\wedge  e_4 \wedge  e_6    \nonumber \\
  |100\rangle \  \longmapsto \   e_2\wedge  e_3 \wedge  e_5  ,
  & \quad &
  |011\rangle \  \longmapsto \   e_1 \wedge  e_4 \wedge  e_6   \nonumber \\
  |010\rangle \  \longmapsto \   e_1\wedge  e_4 \wedge  e_ 5 ,
  & \quad &
  |101\rangle \  \longmapsto \   e_2\wedge  e_3 \wedge  e_6   \nonumber \\
  |001\rangle \  \longmapsto \   e_1\wedge  e_3 \wedge  e_6  ,
  & \quad &
  |110\rangle \  \longmapsto \   e_2\wedge  e_4 \wedge  e_5.    \quad
\label{standard embedding}
\end{eqnarray}
What Borland and Dennis discovered is that every wave function in $\Borland$ is U(6)-equivalent to one in the image of the embedding (\ref{standard embedding}).  

By dint of the Borland-Dennis Theorem, algebraic results for the $3$-qubit setting can be transferred to the Borland-Dennis setting, and {\it vice versa}  \cite{ChenDokovicGrasslZeng}.  
For example, every $3$-qubit wave function is equivalent, under ``local'' unitary transformations of $1$-qubit Hilbert spaces, to one that can be written as a linear combination of $ |000\rangle, |001\rangle,  |010\rangle, |100\rangle$, and $|111\rangle$  \cite{CarteretHiguchiSudbery,AcinAndrianovJaneTarrach 2001}.  
 Analogously, every wave function in $\Borland$ is U(6)-equivalent to a linear combination of the five configurations that correspond via (\ref{standard embedding}) to these five $3$-qubit product states.   
Chen {\it et al.} \cite{ChenDokovicGrasslZeng} have shown how to standardize these $5$-term CI expansions to define a ``canonical form'' 
for representatives of U(6) equivalence classes.

This canonical form for wave functions $\psi \in \Borland$ can be obtained from their ``max-overlap'' Slater determinant approximations.  
A  max-overlap Slater determinant approximation of $\psi$ is a Slater determinant wave function $\chi_* =  f_1\wedge f_2 \wedge  f_3$ such that 
$\langle \chi_*, \psi \rangle$ is greater than or equal to the overlap $ \big|\langle \chi, \psi \rangle\big| $ between $\psi$ and any other Slater determinant $\chi$.  Any max-overlap Slater determinant approximation of $\psi$ is the principal configuration in a $5$-term CI expansion for $\psi$ that is U(6)-equivalent to its canonical form.

Our main result is that every wave function $\psi \in \Borland$ also has a $5$-term CI expansion based upon any of its ``max-overlap CIS approximations.''  


A wave function in the Borland-Dennis space is ``CIS'' if and only if there exists an orthonormal basis $\{f_1,f_2,f_3,g_1,g_2,g_3\}$ of $\Csix$ and coefficients $A,B_1,B_2,B_3$ such that 
\begin{equation}
\label{CIS form}
   \psi \equals
    A \  f_1 \wedge   f_2  \wedge   f_3  \plus B_1 \ g_1  \wedge  f_2 \wedge    f_3  
    \plus  B_2 \ f_1  \wedge   g_2  \wedge   f_3  \plus B_3  \  f_1 \wedge   f_2  \wedge   g_3  \ .
\end{equation}
It turns out that a wave function $\psi \in \Borland$ is CIS if and only if its ``hyperdeterminant'' is $0$.

The structure of the CI expansion (\ref{CIS form}) may be easier to see when it is displayed schematically as follows: 
\begin{equation}
\label{schematically CIS}
\begin{matrix}
 f_1 & f_2 & f_3 & g_1 & g_2 & g_3 & \\
 \mathrm{X} & \mathrm{X} & \mathrm{X} & \mathrm{O} & \mathrm{O} & \mathrm{O} & \qquad A  \\
\mathrm{O} & \mathrm{X} & \mathrm{X} & \mathrm{X} & \mathrm{O} & \mathrm{O} & \qquad B_1  \\
\mathrm{X} & \mathrm{O} & \mathrm{X} & \mathrm{O} & \mathrm{X} & \mathrm{O} & \quad  -B_2  \\
\mathrm{X} & \mathrm{X} & \mathrm{O} & \mathrm{O} & \mathrm{O} & \mathrm{X} & \qquad B_3  \\
\end{matrix}
\end{equation}
 We will often use such ``configuration diagrams'' in this article.
The topmost row of the diagram specifies the orthonormal system of orbitals that configurations are built from.  
Each row beneath this header represents a configuration and a coefficient.
  For example, the second-to-last row denotes the term 
  $ - B_2 \ f_1  \wedge   f_3  \wedge   g_2 = B_2 \ f_1  \wedge   g_2  \wedge   f_3$.  
  The diagram as a whole represents the sum of the terms corresponding to its configuration rows, which may be listed in any order.

Here is our main result.  Suppose $\chi_*$ is a CIS wave function of maximum overlap with $\psi \in \Borland$ with canonical form (\ref{schematically CIS}).
Then there exists a complex number $D$ such that $\psi$ has the CI expansion 
\[
\begin{matrix}
 f_1 & f_2 & f_3 & g_1 & g_2 & g_3 & \\
  \mathrm{X} & \mathrm{X} & \mathrm{X} & \mathrm{O} & \mathrm{O} & \mathrm{O} & \qquad sA  \\
\mathrm{O} & \mathrm{X} & \mathrm{X} & \mathrm{X} & \mathrm{O} & \mathrm{O} & \qquad sB_1  \\
\mathrm{X} & \mathrm{O} & \mathrm{X} & \mathrm{O} & \mathrm{X} & \mathrm{O} & \quad  -sB_2  \\
\mathrm{X} & \mathrm{X} & \mathrm{O} & \mathrm{O} & \mathrm{O} & \mathrm{X} & \qquad sB_3  \\
 \mathrm{O} & \mathrm{O} & \mathrm{O} & \mathrm{X} & \mathrm{X} & \mathrm{X} & \qquad D   \\
\end{matrix}
\]
with $s = \langle \chi_*, \psi \rangle $.  
An analogous result also holds in the $3$-qubit setting.  

\bigskip

The rest of this article is organized as follows:

   \bigskip

  \begin{tabular}{ll}
Section~\ref{Notation, terminology, and lemmas}: \quad & Notation, terminology, and lemmas \\
Section~\ref{Algebraic invariants}: \quad & The algebraic invariants \\
Section~\ref{Small classes of Borland-Dennis wave functions}: \quad & Small classes of Borland-Dennis wave functions \\
Section~\ref{``CIS'' and ``CID'' wave functions}: \quad & ``CIS'' and ``CID'' wave functions \\
Section~\ref{Max-overlap approximations}: \quad & Max-overlap approximations \\
Section~\ref{The Borland-Dennis Theorem}: \quad & The Borland-Dennis Theorem \\
Section~\ref{The 5-term canonical forms}: \quad & The 5-term canonical forms \\
Section~\ref{5-term expansions based on max-overlap CIS approximations}: \quad & 
                     5-term expansions based on max-overlap CIS approximations \\
Appendix:  \quad & Another proof of the Borland-Dennis Theorem \\
\end{tabular}

\bigskip

  Sections~\ref{Notation, terminology, and lemmas} -  \ref{Max-overlap approximations} contain background: definitions, lemmas, and other elements that will be required in the last three sections.   Sections~\ref{The Borland-Dennis Theorem} and \ref{The 5-term canonical forms} review  the Borland-Dennis Theorem and the ``canonical'' $5$-term CI expansions.  Section~\ref{5-term expansions based on max-overlap CIS approximations} presents our main results about $5$-term CI expansions based upon max-overlap CIS approximations and the analog of this result for the $3$-qubit setting.  Finally, a hands-on proof of the Borland-Dennis Theorem is posted in the appendix.


\section{Notation, terminology, and lemmas}
\label{Notation, terminology, and lemmas}

\subsection{Terminology and notation}
\label{Terminology and notation}

To describe $n$-fermion states, we shall use some of the notation and terminology of exterior algebra 
\cite{Cassam-Chenai,Cassam-Chenai_Petras} 
alongside the language of quantum physics and quantum chemistry \cite{ColemanYukalov}.    

When we refer to a vector in any Hilbert space as a ``wave function,'' we mean that it is normalized.   
``Orbitals'' are wave functions in the ``$1$-particle Hilbert space.'' 
In this article, the $1$-particle Hilbert spaces of interest will mainly be isomorphic to $\mathbb{C}^d$ with the standard inner-product, 
and the $n$-fermion Hilbert spaces will be isomorphic to spaces $\wedge^n \mathbb{C}^d$ of alternating $n$-tensors.   
In particular, the Hilbert space for three fermions in six orbitals is isomorphic to $\Borland$.  We often call this the Borland-Dennis space.  

Elements of a vector space $\wedge^n \mathbb{C}^d$ are sometimes called ``multivectors.''  We may call them ``bivectors'' when $n=2$ and  ``trivectors'' when $n=3$.   A multivector in $\wedge^n \mathbb{C}^d$ is ``decomposable'' if it can be written as a simple alternating product $v_1 \wedge  \cdots \wedge v_n$ of vectors in $\mathbb{C}^d$.

Let  $\{e_1,\ldots,e_d\}$ be the standard orthonormal basis of $ \mathbb{C}^d$.  
Hilbert space geometry is imposed on $\wedge^n \mathbb{C}^d$ by declaring the multivectors $e_{i_1} \wedge e_{i_1}  \wedge \cdots \wedge e_{i_n} $ with $1 \le i_1 < i_2 < \cdots < i_n \le d$ to be orthonormal.    
This implies the following general formula for the inner product of two decomposable multivectors:
\[
 \big\langle f_1 \wedge \cdots \wedge f_n, \ g_1 \wedge   \cdots \wedge g_n \big\rangle
  \equals 
   \det \big(\langle f_i, g_j \rangle\big)_{i,j=1}^n.
\]

When $f_1\ldots,f_n$ are orthonormal orbitals in $\mathbb{C}^d$, the decomposable multivector $f_1 \wedge f_2 \wedge \cdots \wedge f_n$ is a unit vector in $\wedge^n \mathbb{C}^d$, called a ``Slater determinant.''
If $(g_1,g_2,\ldots,g_d)$ is an ordered orthonormal basis of  $\mathbb{C}^d$, then the set 
$ \big\{ g_{i_1} \wedge g_{i_2} \wedge \cdots \wedge g_{i_n} :\  1 \le i_1 < i_2 < \ldots <  i_n \le d \big\}   $
of Slater determinants is an orthonormal basis of $\wedge^n\mathbb{C}^d$.   
We will refer the elements of this basis set as ``configurations'' in the ``reference orbitals''  $(g_1,\ldots,g_d)$.    
Any wave function in $\wedge^n \mathbb{C}^d$ can be written as a linear combination of configurations
in the reference orbitals.  This representation of a wave function is known as its ``configuration interaction expansion'' or ``CI expansion'' with respect to those reference orbitals.

Let $\HH$ denote a $1$-particle Hilbert space.  A wave function $\psi \in \wedge^n\HH $ represents a pure $n$-fermion state with density matrix $|\psi \rangle\!\langle \psi |$.   
The state determines, in particular, the answer to the question, ``What is the probability that a given orbital would be found to be occupied, if an experiment to determine its occupation or vacancy were performed?''  These probabilities are the diagonal matrix elements of the $1$-particle reduced density matrix, or ``1RDM'' of $\psi$.  The 1RDM $\Gamma$ of  $\psi $ is the unique Hermitian operator on $\HH$ such that $\langle f, \Gamma  f  \rangle$ is the probability of occupation of the orbital $f$, for any orbital $f$.   The 1RDM can be derived from $\psi $ by   taking a partial trace of the associated density matrix and normalizing it to have trace $n$, that is, 
$$ \Gamma \equals n\  \mathrm{Tr}_{2,\ldots,n} |\psi \rangle\!\langle \psi |.$$ 
Eigenvectors and eigenvalues of the 1RDM derived from $\psi$ are called ``natural orbitals'' of $\psi$ and ``natural occupation numbers'' of $\psi$, respectively.  

We shall use the following notation for subspaces of $\wedge^n \HH$:

If $\VV$ is a subspace of $\HH$, we identify $\wedge^n \VV$ with the subspace 
$  \Span\big\{ v_1 \wedge \cdots \wedge v_n :\ v_1 \wedge \cdots \wedge v_n \in \VV \big\}$
 of $\wedge^n \HH$.  
If $\mathfrak{F}_1$ and $\mathfrak{F}_2$ are orthogonal subspaces of $\wedge^n \HH$, we identify 
$\mathfrak{F}_1 \oplus \mathfrak{F}_2$ and $\mathfrak{F}_1 \wedge \mathfrak{F}_2$ with the subspaces 
\[
      \Span\big\{ \psi_1 + \psi_2  :\ \psi_1 \in \mathfrak{F}_1, \psi_2 \in \mathfrak{F}_2 \big\}
\quad \hbox{and} \quad 
      \Span\big\{ \psi_1 \wedge \psi_2  :\ \psi_1 \in \mathfrak{F}_1, \psi_2 \in \mathfrak{F}_2 \big\}
\] 
of $\wedge^n \HH$, respectively.

Two multivectors $\phi_1 \in \wedge^{n_1} \HH$ and $\phi_2 \in \wedge^{n_2} \HH$ are said to be ``strongly orthogonal'' if there exist orthogonal subspaces $\VV_1$ and $\VV_2$ of $\HH$ such that $\phi_1 \in \wedge^{n_1} \VV_1$ and $\phi_2 \in \wedge^{n_2} \VV_2$.    The ``rank'' of an $n$-fermion wave function $\psi \in \wedge^n \HH$ is the least $m$ such that $\psi$  belongs to a subspace of the form $\wedge^n \MM$ where $\MM$ is an $m$-dimensional subspace of $\HH$.   The rank of a wave function $\psi$ equals the matrix rank of its 1RDM.


\subsection{Basic lemmas}

Let $\psi$ be a wave function in a many-fermion Hilbert space $\wedge^n \HH$.  Orthonormal orbitals $u_1, u_2, \ldots, u_n$ are called ``best-overlap orbitals'' \cite{KutzelniggSmith} for $\psi$ if the overlap $\big| \langle u_1 \wedge  \cdots \wedge  u_n, \psi  \rangle  \big|$ between $\psi$ and the Slater determinant $u_1 \wedge  \cdots \wedge  u_n$ is greater than or equal to the overlap  $\big| \langle \chi, \psi  \rangle \big|$ between $\psi$ and any other Slater determinant $\chi$.

The following lemma expresses the fact that best-overlap orbitals satisfy the ``Brueckner conditions'' \cite{KutzelniggSmith,Nesbet}.    
Though we state and prove the lemma for a Hilbert space $\HH$ with a countably infinite orthonormal basis, the result holds for Hilbert spaces of any dimension.

\begin{Lemma}
\label{Brueckner lemma}
Let $(u_1,u_2,\ldots)$ denote an ordered orthonormal basis of a Hilbert space $\HH$ and suppose that $u_1, u_2, \ldots, u_n$ are best-overlap orbitals for a wave function $\psi \in \wedge^n \HH$.  
 If $\{1,2,\ldots,n\}$ and $\{i_1,\ldots,i_n\}$ have exactly $n-1$ members in common, then  
$ \langle u_{i_1}\wedge  \cdots \wedge  u_{i_n} , \psi \rangle =  0  $.
\end{Lemma}

\begin{proof}
Let 
\[
    \psi \ = \  \sum_{j_1 < \cdots < j_n } C_{j_1 \cdots j_n } u_{j_1} \wedge \cdots \wedge u_{j_n}\ .
\]
be the CI expansion of $\psi$ with respect to the ordered orthonormal basis $(u_1,u_2,\ldots)$.  
As $u_1 \wedge  \cdots \wedge  u_n$ is a best-overlap approximation of $\psi$, the coefficient $C_{1,2,\ldots,n}$ is nonzero; indeed, $|C_{1,2,\ldots,n}|^2$ is the maximum overlap-squared between $\psi$ and any Slater determinant.  

Suppose that $\{i_1,\ldots,i_n\}$ with $i_1< \cdots < i_n $ has exactly $n-1$ members in common with $\{1,2,\ldots,n\}$. 
 Then $i_1< \cdots < i_{n-1} \le n$ and $i_n > n$.  We will prove that $C_{i_1 \cdots i_n } = 0$.  

The multivector $C_{1,\ldots,n} u_1 \wedge  \cdots \wedge  u_n + C_{i_1 \cdots i_n } u_{i_1} \wedge \cdots \wedge u_{i_n}$ is decomposable, because $u_1 \wedge  \cdots \wedge  u_n$ and $u_{i_1} \wedge \cdots \wedge u_{i_n}$ have $n-1$ wedge factors in common, and these may be factored out of the sum,  producing a wedge product of $n$ orbitals.  Let $\chi$ be a (normalized) wave function proportional to this decomposable multivector, so that $\chi$ is a Slater determinant. The overlap-squared $|\langle \chi, \psi \rangle |^2$ equals $|C_{1,\ldots,n} |^2 + |C_{i_1 \cdots i_n } |^2$.  Since $|C_{1,\ldots,n} |^2$ is the maximum possible overlap-squared between $\psi$ and any Slater determinant, $C_{i_1 \cdots i_n } = \langle u_{i_1}\wedge  \cdots \wedge  u_{i_n} , \psi \rangle$ must equal $0$.
\end{proof}

Lemma~\ref{Brueckner lemma} means that the CI expansion of $\psi$ with respect to the ordered orthonormal basis 
$(u_1,u_2,\ldots )$  does not contain any configurations that are ``single excitations'' of $u_1  \wedge  \cdots \wedge  u_n$.
This conclusion can also be restated in a coordinate-free manner.  
Let $\mathcal{R}  = \Span\{u_1, u_2, \ldots, u_n\}$ and let $\mathcal{R}^\perp$ denote its orthogonal complement in $\HH$.  
Lemma~\ref{Brueckner lemma} says that $\psi$ is orthogonal to the subspace $\big( \wedge^{n-1} \mathcal{R}\big) \wedge  \mathcal{R}^\perp$ of $\wedge^n \HH$.

Lemma~\ref{Brueckner lemma} implies the following classical result, for which there are also other, more elementary, proofs   \cite{zjm}.  

\begin{Lemma}\label{hole}
Every multivector in $\wedge^{n}\mathbb{C}^{n+1}$ is decomposable.  
\end{Lemma}

The following lemma follows from Lemma~\ref{Brueckner lemma} \cite{LoewdinShull}.  It can also be proven by using the singular value decomposition  \cite{coleman,PaskauskasYou}.  
\begin{Lemma}
\label{canonical geminal}
Let $\HH$ denote a Hilbert space of dimension at least $2$ and let  $\psi$ be a wave function in $\wedge^2 \HH$.  
There exists an orthonormal basis $\big\{ u_1,v_1, u_2,v_2, \ldots  \big\}$ of $\HH$ and a square summable sequence of nonnegative numbers $A_1 \ge A_2 \ge \cdots $ 
such that 
\begin{equation}
\label{geminal CI}
\psi \equals  \sum_i A_i\  u_i\!\wedge  v_i \ .
\end{equation}
\end{Lemma}

We will frequently use the following two special cases of Lemma~\ref{canonical geminal}:

\begin{Lemma}
\label{two-in-four}
For every bivector $\gamma \in \wedge^2 \mathbb{C}^4$, there exists an orthonormal basis $\{u_1, u_2, v_1,v_2\}$ of $\mathbb{C}^4$ and coefficients  $A_1,A_2$ such that  $\gamma = A_1 u_1 \wedge  v_1 + A_2 u_2 \wedge  v_2$. 
\end{Lemma}

\begin{Lemma}
\label{two-in-five}
For every bivector $\gamma \in \wedge^2 \mathbb{C}^5$, there exist orthonormal orbitals $u_1, u_2, v_1,v_2 \in  \mathbb{C}^5$ and coefficients  $A_1,A_2$ such that  $\gamma = A_1 u_1 \wedge  v_1 + A_2 u_2 \wedge  v_2$. 
\end{Lemma}

The next lemma can be deduced from Lemma~\ref{two-in-five} by using the idea of particle-hole duality, 
but we deduce it directly from Lemma~\ref{Brueckner lemma}.

\begin{Lemma}
\label{three-in-five}
For every trivector $\psi \in \wedge^3  \mathbb{C}^5$, there exists an orthonormal basis $\{w, u_1, u_2, v_1,v_2 \}$ of $ \mathbb{C}^5$ and 
coefficients $A_1, A_2$ such that  
\begin{equation}
\label{canonical for for three-in-five} 
\psi = A_1 \  w \wedge u_1 \wedge  v_1   +   A_2 \  w  \wedge u_2 \wedge  v_2 \ . 
\end{equation}
\end{Lemma}
\begin{proof}
Suppose that $\chi_* = w_1\wedge w_2 \wedge w_3$ is a best-overlap Slater determinant for $\psi$ and let $u_2$ and $v_2$ be two other orbitals such that $\{w_1, w_2 , w_3,u_2,v_2\}$ is an orthonormal basis of $ \mathbb{C}^5$.   Consider the CI expansion of $\psi$ with respect to this basis.  By Lemma~\ref{Brueckner lemma}, the six configurations like  $w_2 \wedge w_3 \wedge u_2$ that are ``singly-excited'' from $\chi_*$ do not appear in the CI expansion, that is, the coefficients of these configurations are $0$.  The CI expansion involves only the configuration $w_1\wedge w_2 \wedge w_3$ and its three double excitations.  These double excitations have the common factor $u_2 \wedge  v_2$, which may be factored out.  Thus, there exists an orbital $w \in \Span\{w_1, w_2 , w_3\}$ and a coefficient $A_2$ such that 
\[
      \psi \equals \langle \chi_* , \psi \rangle \chi_* \plus  A_2 \  w  \wedge u_2 \wedge  v_2 \ .
\]
Let $u_1$ and $v_1$ be orbitals such that $\{ w, u_1, v_1 \}$ is an orthonormal basis of $\Span\{w_1, w_2 , w_3\}$.  
Then there is a complex number $e^{i\theta}$ such that 
$ \chi_* \equals e^{i\theta} \   w \wedge u_1 \wedge  v_1 $,
and $\psi$ is as written in (\ref{canonical for for three-in-five}) with $A_1 = \langle \chi_* , \psi \rangle e^{i\theta}$.  
\end{proof}

The next lemma, which we require in order to prove Lemma~\ref{decomposable in span}, follows from Lemma~\ref{canonical geminal}.  

\begin{Lemma}\label{decomposable geminals}
Let $\HH$ denote a Hilbert space of dimension at least $2$. 
A bivector $\gamma \in \wedge^2 \HH$ is decomposable if and only if $\gamma \wedge  \gamma = 0$. 
\end{Lemma}

\begin{Lemma}\label{decomposable in span}
If $\gamma_1, \gamma_2 \in \wedge^2 \mathbb{C}^4$ are linearly independent, then there is a nonzero decomposable bivector in 
 $\Span\{\gamma_1, \gamma_2\}$. 
\end{Lemma}
\begin{proof}
The case that $\gamma_1$ is decomposable is trivial, so let us assume that $\gamma_1$ is not decomposable. 
Let 
$\vec{e} \equals e_1 \wedge  e_2 \wedge  e_3 \wedge  e_4$ where $\{e_i : 1\leq i \leq 4 \}$ is the standard basis of $\mathbb{C}_4$, and let 
$A,B,C \in \mathbb{C} $ be such that
$\gamma_1\wedge  \gamma_1 = A  \vec{e} $, $\gamma_1\wedge  \gamma_2 = B  \vec{e}$, and $\gamma_2\wedge  \gamma_2 = C  \vec{e}$.
By Lemma~\ref{decomposable geminals}, $\gamma_1 \wedge  \gamma_1 \neq 0$; therefore, $A \neq 0$. 

If $z$ is one of the roots of $Ax^2+2B x+C$, then 
\begin{eqnarray*} 
    (z \gamma_1 + \gamma_2) \wedge  (z \gamma_1 + \gamma_2) = (Az^2+2Bz+C)\vec{e} =0\ ,
\end{eqnarray*}
and $z \gamma_1 + \gamma_2$ is decomposable.  Since $\gamma_1$ and $\gamma_2$ are linearly independent, $z \gamma_1 + \gamma_2$ is nonzero. 
\end{proof}

\subsection{The SVD Lemma}


Let $\RR$ be an $n$-dimensional subspace of a $1$-particle Hilbert space $\HH$ and let $\RR^\perp$ denote its orthogonal complement in $\HH$.  
If $\{f_1,\ldots,f_n\}$ is an orthonormal basis of $\RR$, then the Slater determinant  $f_1 \wedge \cdots \wedge f_n$ is a unit vector whose span is the $1$-dimensional subspace $\wedge^n \RR$ of the $n$-fermion Hilbert space $\wedge^n \HH$.  Slater determinants composed of $n-1$ orbitals from $\{f_1,\ldots,f_n\}$ and $1$ orbital belonging to  $\RR^\perp$ are said to be ``singly-excited'' from the ``reference'' configuration $f_1 \wedge \cdots \wedge f_n$.  
The span of the reference configuration and all of its single-excitations is the subspace 
\begin{equation}
\label{general CIS subspace} 
       \wedge^n \RR \oplus  \big( ( \wedge^{n-1}   \RR ) \wedge   \RR^\perp \big)
\end{equation} 
of $\wedge^n \HH$.  
Wave functions that belong to a subspace like (\ref{general CIS subspace}) are used as variational wave functions in the ``CI singles'' or ``CIS'' method.  We therefore call such wave functions ``CIS'' wave functions.

Any $n$-fermion CIS wave function can be written as a superposition of the reference configuration and at most $n$ single-excitations thereof \cite{Mayer}.    
We prove this only for finite-dimensional $\HH \cong \mathbb{C}^d$ just so that we may apply the well-known singular value decomposition, or SVD, but the fact is true even for infinite-dimensional Hilbert spaces.

\begin{Lemma}[SVD Lemma]
\label{Mayer's lemma}
Let $\RR$ denote a subspace of $\mathbb{C}^d$ with $\dim\RR = n$ and let $\RR^\perp$ denote its orthogonal complement in $\mathbb{C}^d$.  Let $m=\min\{n,d-n\}$ and suppose that $m \ge 1$.   
Suppose a wave function $\psi$ lies in the subspace (\ref{general CIS subspace}) of $\wedge^n \mathbb{C}^d$.  
Then there exists an orthonormal basis $\{f_1,\ldots,f_n\}$ of $\RR$, orthonormal vectors $g_1,\ldots,g_m \in \RR^\perp$, and coefficients 
$A,B_1,\ldots,B_m$, such that 
\begin{equation}
\label{reduced by SVD - Mayer's lemma}
     \psi  \equals A\ f_1 \wedge \cdots \wedge f_n  \plus \sum_{i=1}^m  B_i \ f_1 \wedge \cdots \wedge f_{i-1}\wedge g_i \wedge f_{i+1} \wedge \cdots \wedge f_n  
\end{equation}
\end{Lemma}

\begin{proof}

Let $\{f'_1,\ldots,f'_n\}$ and and $\{g'_1,g'_2,\ldots \}$ be orthonormal bases of $\RR$ and $\RR^\perp$, respectively, and let   
 $$\alpha'_i = (-1)^{n-i} f'_1 \wedge \cdots \wedge f'_{i-1} \wedge f'_{i+1} \wedge \cdots \wedge f'_n\ .$$
There exist coefficients $C_{ik}$ such that 
\begin{equation}
\label{before SVD transformation}
 \psi  \  = \  A\ f'_1 \wedge \cdots \wedge f'_n \plus \sum_{i=1}^n  \sum_{k=1}^{d-n}  C_{ik} \   \alpha'_i  \wedge g'_k  .
\end{equation}

 Let ${\bf U}$ and ${\bf V}$ be unitary matrices for a singular value decomposition of the $n \times (d-n)$ matrix ${\bf  C} = (C_{ik})$.  
 That is, let ${\bf U}= (U_{ij})$ and ${\bf V}= (V_{ij})$ be unitary matrices such that 
 ${\bf U}^* {\bf C} {\bf V} = {\bf \Lambda} $, where ${\bf \Lambda}$ is the $n \times (d-n)$ ``diagonal'' matrix of singular values of ${\bf C}$.  
 Denote the diagonal entries $\Lambda_{kk }$ of $ {\bf \Lambda}$ by $B_k$.  
 Note that the rank of  $ {\bf \Lambda}$ is the same as the rank of ${\bf  C}$, and therefore at most $m$ of the $B_k$ are nonzero.
The matrices ${\bf U}$ and ${\bf V}$ can be chosen to make $\det({\bf U}) = 1$, which we assume for convenience.

 Define the orbitals $f_1,\ldots,f_n$ by the unitary transformation  
\[
  f'_i = \sum_{j=1}^n  U_{ij} f_j\ .
\]
Then $ f'_1 \wedge \cdots \wedge f'_n =  f_1 \wedge \cdots \wedge f_n$ because $\det({\bf U}) = 1$.  

 Define $ \alpha_j = (-1)^{n-i} f_1 \wedge \cdots \wedge f_{j-1} \wedge f_{j+1} \wedge \cdots \wedge f_n $.  
 Then $ \alpha'_i$ can be written as a linear combination of the $\alpha_j$. 
 The coefficient of $\alpha_j$ in this linear combination is the cofactor of the $(i,j)$ entry of the matrix ${\bf U}$, which equals the 
$(j,i)$ entry of the adjugate of ${\bf U}$.  The adjugate of ${\bf U}$ equals ${\bf U}^*$ since ${\bf U}$ is unitary with $\det({\bf U}) = 1$.
Thus 
 \begin{equation}
 \label{rotate the alphas - Mayer's lemma}
  \alpha'_i = \sum_{j=1}^n  {\bf U}^*_{ji }  \alpha_j\ .
\end{equation}

Finally, define $g_1,\ldots,g_{d-n}$ by the unitary transformation  
 \begin{equation}
 \label{rotate the g's - Mayer's lemma}
 \quad g'_k = \sum_{\ell =1}^{d-n} V_{k \ell } g_\ell \ .
\end{equation}
Substituting (\ref{rotate the g's - Mayer's lemma}) and (\ref{rotate the alphas - Mayer's lemma}) into 
(\ref{before SVD transformation}) yields formula (\ref{reduced by SVD - Mayer's lemma}).
\end{proof}

It is convenient to have a statement of the SVD Lemma specifically for the Borland-Dennis setting:

\begin{Lemma}
\label{SVD lemma}
Let $\RR$ and $\RR^\perp$ denote orthogonal $3$-dimensional subspaces of $\Csix$.
Suppose that a wave function $\psi$ lies in the subspace 
\begin{equation}
\label{singles-free subspace}
 \wedge^3 \RR \oplus  \big( \RR \wedge   \RR  \wedge   \RR^\perp) \oplus \wedge^3 \RR^\perp 
\end{equation}
of $\Borland$.  
Then there exist orthonormal bases $\{f_1,f_2,f_3\}$ and $\{g_1,g_2,g_3\}$ of $\RR$ and $\RR^\perp$, respectively, 
and coefficients $A,B_1,B_2,B_3,D$ such that 
\begin{equation}
\label{5-term expansion}
 \psi  \equals   A\   f_1  \wedge  f_2 \wedge    f_3 \plus B_1 \ g_1  \wedge  f_2 \wedge    f_3  \plus  B_2\  f_1  \wedge   g_2  \wedge   f_3  \plus B_3 \  f_1 \wedge   f_2  \wedge   g_3 \plus D\   g_1  \wedge  g_2 \wedge    g_3
\ .
\end{equation}
\end{Lemma}
\begin{proof}
This is a special case of Lemma~\ref{Mayer's lemma}.  One needs only to observe that the transformations in the proof of that lemma leave the $1$-dimensional subspace $\wedge^3 \RR^\perp$ invariant.
\end{proof}


\section{Algebraic invariants}
\label{Algebraic invariants}

Let U(6) denote the group of $6 \times 6$ unitary matrices.  
Each $U \in \mathrm{U(6)}$ induces a unitary operator $\widehat{U}$ on $\Borland$ defined by
\[
     \widehat{U} ( x_1 \wedge x_2 \wedge x_3)  \equals  U x_1 \wedge U x_2 \wedge U x_3
 \]
for all $ x_1, x_2 ,x_3 \in \Csix$.   Thus U(6) acts on the Borland-Dennis space.

A wave function in the Borland-Dennis space has a variety of possible CI expansions, one for each ordered orthonormal basis of $\Csix$.  
The possible CI expansions of a wave function $\psi \in \Borland$ are related to the orbit of $\psi$ under the action of U(6).   
Let $\sum C_{ijk} \ e_i \wedge e_j \wedge e_k $ be the CI expansion of a wave function $\phi \in \Borland$ with respect to the standard basis of $\Csix$.  Then $\phi$ is in the U(6)-orbit of $\psi$ if and only if there exists an ordered orthonormal basis 
$(f_1,\ldots,f_6)$ of $\Csix$ such that $\psi = \sum C_{ijk} \ f_i \wedge f_j \wedge f_k $.

Chen, \DJ okovi\'c, Grassl, and Zeng \cite{ChenDokovicGrasslZeng} 
 have presented a complete set of algebraic invariants for the action of U(6) on the Borland-Dennis space. 
These invariants are polynomial functions of CI expansion coefficients that take the same value for for all possible choices of reference orbitals (the coefficients themselves change but certain polynomial functions of them are invariant).   The polynomial algebra of invariants is generated by six independent invariants $M_1,\ldots,M_6$ and one auxiliary invariant, $M_7$, related to the others by a syzygy.  

The algebraic invariant $M_1$ is numerically equal to the norm squared of the trivector.  For normalized  $\psi \in \Borland$ with natural occupation numbers 
$\lambda_1 \ge \cdots \ge \lambda_6$, the algebraic invariants  $M_2,M_4$ and $M_6$ are numerically equal to the elementary symmetric polynomials in $\lambda_1\lambda_6$, $\lambda_2\lambda_5$, and $\lambda_3\lambda_4$.  The invariant $M_3$ depends on the 2RDM of $\psi$.   
A useful combination of these basic invariants is $M_1M_2 - 2M_3$, 
for this is non-negative and equals $0$ at a wave function $\psi$ if and only if $\psi$ is not of full rank, i.e., has rank $5$ or less \cite{ChenDokovicGrasslZeng-footnote}.

The last independent invariant is related to Cayley's hyperdeterminant for $2\times 2 \times 2$ matrices.  
It was rediscovered by L\'evay and Vrana \cite{LevayVrana}, who motivated it as a generalization of the $3$-tangle \cite{CoffmanKunduWootters}, a measure of multipartite entanglement for $3$-qubit states, and used it to distinguish between the major equivalence classes for the action of GL(6) on the Borland-Dennis space.  
The invariant, denoted $M_5$ in Ref.~\cite{ChenDokovicGrasslZeng}, is non-negative, and it is positive at a wave function $\psi$ if and only if  $\psi$ belongs to the generic GL(6) equivalence class.  The invariant $M_5$ is the modulus-squared of the hyperdeterminant invariant denoted by $F$ in Ref.~\cite{ChenDokovicGrasslZeng} or $\mathcal{D}$ in Refs.~\cite{SarosiLevay - classification,SarosiLevay - CKW}.  
When we refer to the hyperdeterminant of a wave function $\psi$, we mean $F(\psi)$ or $\mathcal{D}(\psi)$.

The U(6) invariants for the Borland-Dennis setting are very closely related \cite[Sec. 5]{ChenDokovicGrasslZeng} to the local unitary invariants for the symmetrized $3$-qubit setting \cite{Sudbery,AcinAndrianovJaneTarrach 2001}.   

In this paper we are interested in CI expansions that can be obtained from one another by U(6) changes-of-basis, and accordingly the U(6) equivalence classes are of primary interest.  
To complete the picture, we should touch upon the effect of allowing general changes of $1$-particle basis, and mention the equivalence classes of the action on 
$\Borland$ of GL(6,$\mathbb{C}$), the group of all invertible $6 \times 6$ complex matrices.   
By means of a general change of $1$-particle basis, any nonzero trivector in $\Borland$ can be written \cite{Gurevich} in one of the following four canonical forms:
\begin{trivlist}
\item{(i)}\quad 
$ v_1 \wedge  v_2 \wedge  v_3$ 
\item{(ii)}\quad
$ v_1 \wedge  (v_2 \wedge  v_3 + v_4 \wedge  v_5 )$ 
\item{(iii)}\quad
 $ v_1 \wedge  v_2 \wedge  v_3 + v_3 \wedge  v_4 \wedge  v_5 + v_5 \wedge  v_6 \wedge  v_1$ 
\item{(iv)}\quad
$ v_1 \wedge  v_2 \wedge  v_3 + v_4 \wedge  v_5 \wedge  v_6$
\end{trivlist}
where the vectors $v_1,\ldots,v_6$ are linearly independent but not necessarily orthonormal.  
These four GL(6,$\mathbb{C}$) equivalence classes portray -- it is professed -- the essentially different kinds of ``entanglement'' that a pure state of $3$ fermions in $6$ orbitals could have  \cite{LevayVrana}, and are now called, respectively, the ``separable,'' ``biseparable,'' ``W,'' and ``GHZ'' classes \cite{SarosiLevay - classification}.  

Generic trivectors in $\Borland$ have the canonical form (iv). 
Thus, every wave function of rank $6$ or less can be written as a linear combination of at most three Slater determinants.  
However, these are not CI expansions, because the Slater determinants involved need not be built from a system of orthonormal orbitals.  
Generic wave functions in $\Borland$ do not even have $4$-term CI expansions, much less $3$-term CI expansions \cite[Prop.~10]{ChenChenDokovicZeng}.

\section{Small classes of Borland-Dennis wave functions}
\label{Small classes of Borland-Dennis wave functions}

\subsection{Types of CI expansions with 3 or fewer terms}

Using Lemmas~\ref{hole}, \ref{two-in-four}, and \ref{three-in-five}, 
it is not difficult to show that a wave function that has a CI expansion with three or fewer configurations is U(6)-equivalent to one of the following five types.  In labeling these types we follow Refs.~\cite{AcinAndrianovJaneTarrach 2000} and \cite{AcinAndrianovJaneTarrach 2001}, where the corresponding $3$-qubit states are so classified:
\begin{trivlist}
\item{\bf Type 1}
\[
\begin{matrix}
\mathrm{X} & \mathrm{X} & \mathrm{X} & \mathrm{O} & \mathrm{O} & \mathrm{O} & \quad A   \\
\end{matrix}
\]
\item{\bf Type 2a}
\[
\begin{matrix}
\mathrm{X} & \mathrm{X} & \mathrm{X} & \mathrm{O} & \mathrm{O} & \mathrm{O} & \quad A_1  \\
\mathrm{X} & \mathrm{O} & \mathrm{O} & \mathrm{X} & \mathrm{X} & \mathrm{O} & \quad A_2   \\
\end{matrix}
\]
\item{\bf Type 2b}
\[
\begin{matrix}
\mathrm{X} & \mathrm{X} & \mathrm{X} & \mathrm{O} & \mathrm{O} & \mathrm{O} & \quad A   \\
\mathrm{O} & \mathrm{O} & \mathrm{O} & \mathrm{X} & \mathrm{X} & \mathrm{X} & \quad C   \\
\end{matrix}
\]
\item{\bf Type 3a}
\[
\begin{matrix}
\mathrm{O} & \mathrm{X} & \mathrm{X} & \mathrm{X} & \mathrm{O} & \mathrm{O} & \quad B_1   \\
\mathrm{X} & \mathrm{O} & \mathrm{X} & \mathrm{O} & \mathrm{X} & \mathrm{O} & \quad B_2   \\
\mathrm{X} & \mathrm{X} & \mathrm{O} & \mathrm{O} & \mathrm{O} & \mathrm{X} & \quad B_3  \\
\end{matrix}
\]
\item{\bf Type 3b}
\[
\begin{matrix}
\mathrm{X} & \mathrm{X} & \mathrm{X} & \mathrm{O} & \mathrm{O} & \mathrm{O} & \quad A   \\
\mathrm{O} & \mathrm{O} & \mathrm{X} & \mathrm{X} & \mathrm{X} & \mathrm{O} & \quad B  \\
\mathrm{O} & \mathrm{O} & \mathrm{O} & \mathrm{X} & \mathrm{X} & \mathrm{X} & \quad C   \\
\end{matrix}
\]
\end{trivlist}

\subsection{Corresponding classes of wave functions}

To each of the preceding five types of CI expansions there corresponds a class of wave functions, i.e., the set of wave functions that have CI expansions of that type with respect to some basis of reference orbitals.  We use special names for the classes of wave functions that correspond to the first four types:

\begin{trivlist}
 \item{Type 1 -- \ {\bf Slater determinant}\ :}

A wave function has a CI expansion of Type~1 if and only if it is a Slater determinant.  

 \item{Type 2a --  {\bf Low-rank}\ :}

Lemma~\ref{three-in-five} implies that a wave function has a CI expansion of Type~2a if and only if it has rank $5$ or less.
By definition, a trivector $\psi \in \Borland$ has rank $5$ or less if and only there exists a $5$-dimensional subspace $\MM$ of $\Csix$  such that $\psi \in \wedge^3 \MM$. 
In this paper we call wave functions of rank $5$ or less ``low-rank.''   

The rank of a ``low-rank'' wave function is either $5$ or $3$, because a trivector in $\Borland$ cannot have rank $4$ by Lemma~\ref{hole}.   
Lemmas~\ref{two-in-four} and \ref{three-in-five} imply that a trivector $\psi \in \Borland$ has rank  $5$ or less if and only if there exists an orbital $w$ and a geminal $\gamma$, strongly orthogonal to $w$, such that $\psi = w \wedge \gamma$.  
In other words, a trivector $\psi \in \Borland$ has rank $5$ or less if and only there exists a $1$-dimensional subspace $\WW$ of $\Csix$  such that $\psi \in \WW \wedge \big( \wedge^2 \WW^\perp \big)$.


We shall call wave functions that have CI expansions of Type~2b  ``ortho-GHZ'' because they are associated, via an embedding like (\ref{standard embedding}), to $3$-qubit states called ``generalized GHZ states'' in Ref.~\cite{AcinAndrianovJaneTarrach 2000}.

\item{Type 3a -- {\bf Ortho-W}\ :}

Wave functions that have CI expansions of Type~3a we shall call ``ortho-W.''   

In other words, a wave function $\psi$ is ortho-W if and only if there exists an orthonormal basis of reference orbitals $\{f_1,f_2,f_3,g_1,g_2,g_3\}$ of $\Csix$ and expansion coefficients $B_1,B_2,B_3$ such that 
\begin{equation}
\label{write out the ortho-W}
 \psi \equals B_1\   g_1  \wedge  f_2 \wedge f_3  +  B_2\  f_1  \wedge  g_2  \wedge  f_3 + B_3\   f_1  \wedge  f_2  \wedge  g_3 \ .
\end{equation}

Let $\RR = \Span\{f_1,f_2,f_3\}$ and $\RR^\perp = \Span\{g_1,g_2,g_3\}$.  The wave function (\ref{write out the ortho-W}) belongs to the subspace  $\RR \wedge \RR  \wedge   \RR^\perp$.  
Conversely, suppose that a wave function $\psi \in \Borland$ belongs to a subspace of the form  $\RR \wedge \RR  \wedge   \RR^\perp$, where 
$\RR$ and $\RR^\perp$ are orthogonal $3$-dimensional subspaces of $\Csix$.   Then Lemma~\ref{SVD lemma} implies that $\psi$ has a CI expansion (\ref{write out the ortho-W}).    

Thus we have an equivalent characterization of the class of ortho-W wave functions:  a wave function is ortho-W if and only if it belongs to a subspace of $\Borland$ of the form $\RR \wedge \RR  \wedge   \RR^\perp$ for some $\RR \subset \Csix$ with $\dim\RR = 3$.  

\item{Type 3b \ : }

Though we don't have a special term for wave functions with CI expansions of Type~3b, these do constitute a distinct class.  
A wave function $\psi$ that has a CI expansion of Type~3b  with $ABC \neq 0$ cannot have a CI expansion of one of the other types.  
It can't be a low-rank wave function because it has full rank.  It can't be ortho-GHZ because -- one can show -- it has at least three distinct natural occupation numbers, whereas an ortho-GHZ wave function has only two.  
Finally, it can't be ortho-W because its hyperdeterminant equals $A^2C^2 \neq 0$, while the hyperdeterminant of an ortho-W wave function equals $0$.

\end{trivlist}


\section{``CIS'' and ``CID'' wave functions}
\label{``CIS'' and ``CID'' wave functions}

For each three dimensional subspace $\RR$ of $\Csix$,  
the $20$-dimensional space $\Borland$ 
equals the direct sum 
\[
  \wedge^3 \RR \oplus \big(\RR \wedge  \RR  \wedge  \RR^\perp \big)  \oplus \big( \RR \wedge   \RR^\perp  \wedge   \RR^\perp \big)\oplus \wedge^3 \RR^\perp 
\]
of subspaces of dimension $1$, $9$, $9$, and $1$.  

 Suppose that $f_1,f_2,f_3$ and $g_1,g_2,g_3$ are orthonormal orbitals that span $\RR$ and $\RR^\perp$, respectively.  
 The ``reference configuration''  $f_1 \wedge f_2 \wedge f_3$ spans the $1$-dimensional subspace $\wedge^3\RR$ and $g_1 \wedge g_2 \wedge g_3$ spans $\wedge^3\RR^\perp$.   The $9$-dimensional subspace $\RR \wedge  \RR  \wedge  \RR^\perp$ is spanned by configurations like $f_1 \wedge f_2 \wedge g_2$ that are ``singly-excited'' from the ``reference'' configuration $f_1 \wedge f_2 \wedge f_3$.   Similarly, the space $\RR \wedge  \RR^\perp  \wedge  \RR^\perp$ is spanned by configurations like $f_1 \wedge g_1 \wedge g_3$ that are ``doubly-excited'' from the reference configuration.

Wave functions in the subspace 
\begin{equation}
\label{CIS subspace} 
       \wedge^3 \RR \oplus  \big( \RR \wedge   \RR  \wedge   \RR^\perp \big)
\end{equation} are linear combinations of the 
reference configuration $f_1 \wedge f_2 \wedge f_3$ and all singly-excited configurations, 
like a variational wave function for the ``CI singles'' or ``CIS'' method.

Similarly, a wave function in the subspace 
\begin{equation}
\label{CID subspace}
      \wedge^3 \RR \oplus  \big( \RR \wedge   \RR^\perp  \wedge   \RR^\perp \big)
\end{equation}
 is of a ``CI doubles'' or ``CID'' form, as it is a linear combination of the reference configuration and ``double excitations'' thereof.

\begin{Definition}
\label{definition of CIS and CID}
A wave function in $\psi \in \Borland$ is ``CIS''  if there exists a $3$-dimensional subspace $\RR$ of $\Csix$ such that $\psi$ lies in the subspace (\ref{CIS subspace}) of $\Borland$.   A wave function in $\psi \in \Borland$ is ``CID''  if there exists a $3$-dimensional subspace $\RR$ of $\Csix$ such that $\psi$ lies in the subspace (\ref{CID subspace}) of $\Borland$.  
A wave function that belongs to (\ref{CIS subspace}) or (\ref{CID subspace}) is said to be CIS or CID with ``reference space'' $\RR$.
\end{Definition}

In Corollary~\ref{hyperdet = 0 iff CIS} we will prove that a wave function is CIS if and only if its hyperdeterminant equals $0$.  
Using this fact, together with the canonical forms for CIS and CID wave functions described in the next section, one can prove that 
the intersection of the CIS and CID classes is the class of ortho-W wave functions. 

\subsection{Canonical forms}
\label{Canonical forms}

Lemma~\ref{SVD lemma} tells us the following:

Suppose that a wave function $\psi \in \Borland$ belongs to a subspace 
$\wedge^3 \RR \oplus  \big( \RR \wedge   \RR  \wedge   \RR^\perp) \oplus \wedge^3 \RR^\perp$, where $\RR$ is a $3$-dimensional subspace of $\Csix$.
Then there exist orthonormal bases $\{f_1,f_2,f_3\}$ and $\{g_1,g_2,g_3\}$ of $\RR$ and $\RR^\perp$, respectively, 
and coefficients $A,B_1,B_2,B_3,D$ such that 
$\psi$ has CI expansion 
\begin{equation}
\label{schematically 5-term}
\begin{matrix}
 f_1 & f_2 & f_3 & g_1 & g_2 & g_3 & \\
 \mathrm{X} & \mathrm{X} & \mathrm{X} & \mathrm{O} & \mathrm{O} & \mathrm{O} & \qquad A  \\
\mathrm{O} & \mathrm{X} & \mathrm{X} & \mathrm{X} & \mathrm{O} & \mathrm{O} & \qquad B_1  \\
\mathrm{X} & \mathrm{O} & \mathrm{X} & \mathrm{O} & \mathrm{X} & \mathrm{O} & \quad  -B_2  \\
\mathrm{X} & \mathrm{X} & \mathrm{O} & \mathrm{O} & \mathrm{O} & \mathrm{X} & \qquad B_3  \\
 \mathrm{O} & \mathrm{O} & \mathrm{O} & \mathrm{X} & \mathrm{X} & \mathrm{X} & \qquad D  \\
\end{matrix}
\end{equation}

\subsubsection{Canonical form of CIS wave functions}

A CIS wave function with reference space $\RR$ can be represented by 
the first four rows of the configuration diagram (\ref{schematically 5-term}).  The phases of the reference orbitals can be adjusted and their order permuted to make $A \ge 0$ and $B_1 \ge B_2 \ge B_3 \ge 0$.
Thus, a wave function $\psi$ is CIS if and only if there exist coefficients 
$A \ge 0, B_1 \ge B_2 \ge B_3 \ge 0$ and an orthonormal system of reference orbitals $\{ f_1 , f_2 , f_3 , g_1 , g_2 , g_3\}$ such that  $\psi$ has CI expansion 
\begin{equation}
\label{canonical form CIS}
\begin{matrix}
 f_1 & f_2 & f_3 & g_1 & g_2 & g_3 & \\
 \mathrm{X} & \mathrm{X} & \mathrm{X} & \mathrm{O} & \mathrm{O} & \mathrm{O} & \qquad A  \\
\mathrm{O} & \mathrm{X} & \mathrm{X} & \mathrm{X} & \mathrm{O} & \mathrm{O} & \qquad B_1  \\
\mathrm{X} & \mathrm{O} & \mathrm{X} & \mathrm{O} & \mathrm{X} & \mathrm{O} & \quad  -B_2  \\
\mathrm{X} & \mathrm{X} & \mathrm{O} & \mathrm{O} & \mathrm{O} & \mathrm{X} & \qquad B_3 .
\end{matrix}
\end{equation}   
We call (\ref{canonical form CIS}) a ``canonical'' form because the coefficients are unique, provided that $\psi$ has full rank.  
That is, if $\psi$ is a rank-$6$ wave function with CI expansion (\ref{canonical form CIS}), and if $\psi$ can also be written in the same form with respect to reference orbitals $f'_1,\ldots,g'_3$ and coefficients $A' \ge 0$ and $B'_1 \ge B'_2 \ge B'_3 > 0$, then $A' = A$ and $B'_i = B_i$ for $i=1,2,3$.  
This can be shown using formulas (45) - (47) in Ref.~\cite{ChenDokovicGrasslZeng} for the algebraic invariants $M_2, M_3$, and $M_4$. 

A low-rank wave function is a degenerate CIS wave function.  Most low-rank wave functions have infinitely many CI expansions of the form (\ref{canonical form CIS}).   Lemma~\ref{three-in-five} shows that the natural occupation numbers $\lambda_1 \ge \lambda_2 \ge \lambda_3 \ge 1/2$ of a rank-$5$ wave function satisfy $\lambda_1 = 1$ and $\lambda_2 = \lambda_3$.   Let $\lambda$ denote the common value of $\lambda_2$ and $\lambda_2$ and let $\mu = \sqrt{\lambda}$ and $\nu = \sqrt{1 - \lambda}$.  The same rank-$5$ wave function can be written in the form (\ref{canonical form CIS}) in one way for each value of $A$ ranging from $0$ up to $\sqrt{1-2\mu\nu}$.  When $A = 0$,  $B_1 = \mu$, $B_2 = \nu$, and $B_3 = 0$.  At the other extreme, where $A = \sqrt{1-2\mu\nu}$, $B_1 = B_2 = \mu\nu$ and $B_3 = 0$.

\subsubsection{Canonical form of CID wave functions}

A CID wave function with reference space $\RR^\perp$ belongs to the subspace $\big( \RR \wedge   \RR  \wedge   \RR^\perp) \oplus \wedge^3 \RR^\perp$ and therefore can be represented by 
the last four rows of the configuration diagram (\ref{schematically 5-term}).  
Renaming the reference orbitals and adjusting their phases gives CID wave functions a canonical form:   
a wave function $\psi$ is CID if and only if there exists there exists an orthonormal system of reference orbitals $\{ f_1 ,\ldots, f_6\}$ and coefficients 
$A_1 \ge A_2 \ge A_3 \ge A_4 \ge 0$ such that  $\psi$ has CI expansion 
 \begin{equation}
 \label{canonical form CID}
 \begin{matrix}
 f_1 & f_2 & f_3  & f_4 & f_5 & f_6 & \\
 \mathrm{X} & \mathrm{X} & \mathrm{X} & \mathrm{O} & \mathrm{O} & \mathrm{O} & \qquad A_1  \\
 \mathrm{X} & \mathrm{O} & \mathrm{O} & \mathrm{O} & \mathrm{X} & \mathrm{X} & \qquad A_2    \\
 \mathrm{O} & \mathrm{X} & \mathrm{O} & \mathrm{X} & \mathrm{O} & \mathrm{X} & \qquad A_3   \\
 \mathrm{O} & \mathrm{O} & \mathrm{X} & \mathrm{X} & \mathrm{X} & \mathrm{O} & \qquad A_4  .  
 \end{matrix}
\end{equation}
The orbitals $f_1 ,\ldots, f_6$ in (\ref{canonical form CID}) are obviously natural orbitals of $\psi$.  
Therefore, for generic CID wave functions with distinct natural occupation numbers, the canonical expansion (\ref{canonical form CID}) is essentially unique, that is, the orbitals $f_i$ are unique up to multiplication by phase factors and the coefficients $A_i$ are unique.     
In any case, at least the coefficients are unique \cite{ChenDokovicGrasslZeng}.  

\section{Max-overlap approximations}
\label{Max-overlap approximations}

Think of $\psi$ as a target for approximation by wave functions $\chi$, with the aim of maximizing the overlap-squared 
$ \big|\langle \chi,\psi \rangle\big|^2$ over all $\chi$ of some prescribed type.  
For example, $\chi$ may be restricted to be a Slater determinant wave function, or it may be allowed to range over a more general class, such as all low-rank wave functions, or all CIS wave functions.   A wave function [in class X] that maximizes the overlap with a target wave function will be called a ``max-overlap [class X] approximation.''  
To avoid the ambiguity associated to multiplication by a phase factor, we will always choose the phase of a max-overlap approximation such that its overlap with the target wave function is positive.   

Note that max-overlap approximations need not be unique.  
For example, both 
$f_1  \wedge  f_2 \wedge f_3$ and $ f_4  \wedge  f_5  \wedge  f_6$ are max-overlap Slater determinant approximations of the wave function 
$$ \tfrac{1}{\sqrt{2}}  f_1  \wedge  f_2 \wedge f_3  +  \tfrac{1}{\sqrt{2}}  f_4  \wedge  f_5  \wedge  f_6 ,$$
and the wave function 
$$ \tfrac{1}{\sqrt{3}}  g_1  \wedge  f_2 \wedge f_3  +  \tfrac{1}{\sqrt{3}}  f_1  \wedge  g_2  \wedge  f_3 +  \tfrac{1}{\sqrt{3}}  f_1  \wedge  f_2  \wedge  g_3  $$
has infinitely many max-overlap Slater determinant approximations. 

Max-overlap Slater determinant approximations for wave functions of $n$ fermions in $d$ orbitals are not easy to find in general.  
The best general algorithm we know of is iterative, and increases overlap with each iteration, but sometimes gets stuck at a local maximum \cite{zjm}.  Perhaps a better method could be developed specially for the Borland-Dennis setting. 
Indeed, an effective way to get the {\it value} of the maximum, though not the maximizers themselves, has been proposed in Ref.~\cite{ChenDokovicGrasslZeng}.   The maximum overlap-squared is one of the roots of an eighth degree polynomial which, printed out, fills four pages.   

The simplest kind of max-overlap approximation is where $\chi$ is restricted to belong to a fixed linear subspace of $ \Borland$.  
Then the max-overlap approximation is proportional to the projection of the target wave function onto the subspace.  This is a consequence of the following lemma:

\begin{Lemma}
\label{useful little observation}
Let $\mathfrak{H}_0$ be a closed subspace of a Hilbert space $\mathfrak{H}$.
Let $\psi$ be a vector in $\mathfrak{H}$ and let $\phi$ denote its projection onto $\mathfrak{H}_0$.  
Then 
\begin{equation}
\label{useful little maximization}
     \max \Big\{  \big|\langle  \chi, \psi \rangle\big|^2  : \ \chi \in \mathfrak{H}_0, \|\chi\| = 1 \Big\} \equals \| \phi \|^2\ .
\end{equation}
Unless $\psi$ is orthogonal to $\mathfrak{H}_0$, the maximum (\ref{useful little maximization}) is attained only at the unit vectors in $\mathfrak{H}_0$ that are proportional to $\phi$.
\end{Lemma}

\begin{proof}
Since $\phi$ is the projection of $\psi$ onto $\mathfrak{H}_0$,  $ \langle \chi ,\psi  \rangle  =  \langle \chi ,\phi  \rangle $ for any unit vector $\chi \in \mathfrak{H}_0$.  
Therefore, 
\[
    \big|\langle  \chi, \psi \rangle\big|^2 \equals  \big|\langle  \chi, \phi \rangle\big|^2 \ \le\  \| \chi \|^2 \| \phi \|^2  \equals \| \phi \|^2 
\]
by the Cauchy-Schwarz Inequality.  In case $\phi$ is nonzero, equality holds if and only if $\chi$ and $\phi$ are parallel.  
\end{proof}

Lemma~\ref{useful little observation} tells us that finding a max-overlap low-rank approximation of a target wave function $\psi$ is equivalent to finding a $5$-dimensional linear subspace $\MM$ such that the projection of $\psi $ onto the subspace $\wedge^3 \MM$ is as long as possible.    These subspaces can all be found:

\begin{Proposition}
\label{max-overlap rank-5 proposition}
Suppose that $\chi_* \in \Borland$ is a max-overlap low-rank approximation of a wave function $\psi$ with natural occupation numbers $\lambda_1 \ge \lambda_2 \ge \cdots \ge \lambda_6$.  
Then $\langle  \chi_*, \psi \rangle^2 = \lambda_1$.  Moreover, there exist orthonormal orbitals $w$ and $w'$, and two bivectors $\gamma,\gamma'$ strongly orthogonal to them, such that 
\begin{trivlist}
\item{(i)}  \  $w$ and $w'$ are natural orbitals of $\psi$ with natural occupation numbers $\lambda_1$ and $\lambda_6$, respectively,
\item{(ii)}   \  $ \psi = w \wedge \gamma + w' \wedge \gamma'$, 
\item{(iii)}   \  and $w \wedge \gamma =  \langle  \chi_*, \psi \rangle \ \chi_*$.
\end{trivlist}
\end{Proposition}
\begin{proof}
Suppose $\chi_*$ is a max-overlap low-rank approximation of $\psi$.  
By Lemma~\ref{three-in-five}, there exists an orthonormal basis $\{w,u_1,u_2,v_1,v_2,w'\}$ of $\Csix$ and coefficients $A_1,A_2$ such that 
$$\chi_* = w \wedge \big( A_1 u_1 \wedge v_1 + A_2 u_2 \wedge v_2 \big)\ .$$  
Let  $\LL = \Span\{u_1,u_2,v_1,v_2\}$.  There exist 
multivectors $\gamma,\gamma' \in \wedge^2 \LL$ and $\beta \in \wedge^3\LL$, and an orbital $l \in \LL$, such that 
\begin{equation}
\label{full psi}
   \psi \equals w \wedge \gamma \plus w' \wedge \gamma'  \plus w \wedge w' \wedge l \plus \beta \ . 
\end{equation}
We will prove that the last two terms vanish.

Let $\WW = \Span\{w\}$,  $\MM = \Span\{w,u_1,u_2,v_1,v_2\}$ and $\MM' = \Span\{u_1,u_2,v_1,v_2,w'\}$.  
Then $\chi_*$ belongs to both $\wedge^3 \MM$ and $\WW\wedge\big( \wedge^2 \MM' \big)$.   

Since every wave function in $\wedge^3 \MM$ is low-rank, and $\chi_*$ is a max-overlap low-rank wave function, $\chi_*$ must be proportional to the projection of $\psi$ onto $\wedge^3 \MM$, by Lemma~\ref{useful little observation}.   This projection is $w \wedge \gamma + \beta$, from (\ref{full psi}).   But  $\chi_* = w \wedge \big( A_1 u_1 \wedge v_1 + A_2 u_2 \wedge v_2 \big)$ is orthogonal to $\beta$.   It follows that $\beta = 0$.  Again, $\chi_* = w \wedge \big( A_1 u_1 \wedge v_1 + A_2 u_2 \wedge v_2 \big)$ is proportional to the projection of $\psi$ onto $\wedge^3 \MM$, which is just $w \wedge \gamma$.  It follows that (iii) holds, and in fact $\gamma = \langle \chi_*, \psi \rangle  \big( A_1 u_1 \wedge v_1 + A_2 u_2 \wedge v_2 \big)$.

Similarly, since every wave function in $\WW\wedge\big( \wedge^2 \MM' \big)$ is low-rank, 
 $\chi_*$ must also be proportional to the projection of $\psi$ onto this subspace, which is $w \wedge \gamma + w \wedge w' \wedge l$.  But $\chi_*$ is orthogonal to  $w \wedge w' \wedge l$, so the latter trivector must also be $0$.  Thus $\psi$ is just as in (ii).

Let $\Gamma$ denote the 1RDM of $\psi$.  The formula in (ii) shows that $\langle w , \Gamma w \rangle=\big|\langle  \chi_*, \psi \rangle\big|^2$.  This diagonal matrix element of $\Gamma$ is less than or equal to $\lambda_1$, the greatest eigenvalue of $\Gamma$, with equality if and only if $w$ is a corresponding eigenvector of  $\Gamma$, i.e., a natural orbital of $\psi$ with occupation number $\lambda_1$.  Since  $\chi_*$ has maximum overlap with $\psi$ among all low-rank wave functions, $w$ has to be a natural orbital with occupation number $\lambda_1$, and 
$\langle  \chi_*, \psi \rangle^2  = \lambda_1$.  
Since $w$ is a natural orbital, the off-diagonal matrix elements $\langle w , \Gamma w' \rangle$ and $\langle w' , \Gamma w \rangle$ the 1RDM must vanish.  Furthermore, it can be seen from formula (ii) for $\psi$ that the other off-diagonal matrix elements connecting $w'$ to $u_1,u_2,v_1,$ and $v_2$ must also vanish, and that $\langle w' , \Gamma w' \rangle = 1 - \lambda_1 = \lambda_6$.  Thus $w'$ is a natural orbital with occupation number $\lambda_6$.  This proves (i) and the fact that the maximum overlap equals $\lambda_1$.
\end{proof}


\section{The Borland-Dennis Theorem}
\label{The Borland-Dennis Theorem}

In 1970 Borland and Dennis \cite{borland1} published their observation that wave functions in $\Borland$ 
can be expressed compactly in terms of their natural orbitals:  

\begin{Theorem}
\label{Borland-Dennis Theorem}
If $\psi$ is a wave function in $\Borland$, then there exists an ordered orthonormal system $(h_1,\ldots,h_6)$ of its natural orbitals such that 
$\psi$ has CI expansion 
\begin{equation}
\label{eight}
\begin{matrix}
 h_1 & h_2 & h_3 & h_4 & h_5 & h_6  \\
\mathrm{X} & \mathrm{X} & \mathrm{X} & \mathrm{O} & \mathrm{O} & \mathrm{O} & \quad A_{123}   \\
\mathrm{X} & \mathrm{X} & \mathrm{O} & \mathrm{X} & \mathrm{O} & \mathrm{O} & \quad A_{124}  \\
\mathrm{X} & \mathrm{O} & \mathrm{X} & \mathrm{O} & \mathrm{X} & \mathrm{O} & \quad  A_{135}  \\
\mathrm{X} & \mathrm{O} & \mathrm{O} & \mathrm{X} & \mathrm{X} & \mathrm{O} & \quad A_{145}  \\
\mathrm{O} & \mathrm{X} & \mathrm{X} & \mathrm{O} & \mathrm{O} & \mathrm{X} & \quad A_{236}  \\
\mathrm{O} & \mathrm{X} & \mathrm{O} & \mathrm{X} & \mathrm{O} & \mathrm{X} & \quad A_{246}  \\
\mathrm{O} & \mathrm{O} & \mathrm{X} & \mathrm{O} & \mathrm{X} & \mathrm{X} & \quad A_{356}  \\
\mathrm{O} & \mathrm{O} & \mathrm{O} & \mathrm{X} & \mathrm{X} & \mathrm{X} & \quad A_{456}   \\
\end{matrix}
\end{equation}
\end{Theorem}

Note that each of the $8=2^3$ configurations in (\ref{eight}) takes exactly one orbital from each of the three pairs $\{h_1,h_6\}$, $\{h_2,h_5\}$, and $\{h_3,h_4\}$.  The fact that wave functions in $\Borland$ have this form is a special case of a remarkable theorem in Ref.~\cite{ChenChenDokovicZeng}, which is worth repeating here.  
{\it Every wave function in $\wedge^3 \mathbb{C}^d$ with $d \ge 6$ can be written in a ``single-occupancy vector'' form relative to a pairing of orbitals of $\mathbb{C}^d$, wherein configurations take {\it at most} one orbital from each pair.}   In the case where $d=6$ there are only three pairs of orbitals, and configurations must take {\it exactly} one orbital from each pair.  

\subsection{Significance of the Borland-Dennis Theorem}

Using Theorem~\ref{Borland-Dennis Theorem}, Borland and Dennis went on to characterize the set of 
$1$-particle reduced density matrices (1RDMs) that can derived from $3$-fermion wave functions of rank $6$ or less \cite{borland2}.  
This discovery of Borland and Dennis ``stood for more than three decades as the only known solution of the $N$-representability problem beyond two electrons [or] two holes'' \cite{turkey}.  
Here is the result:
\begin{Corollary}
\label{3-representability}
A $6 \times 6$ Hermitian matrix is the 1RDM of a wave function $\psi \in \wedge^3 \mathbb{C}^6$ if and only if 
its eigenvalues $\lambda_i$ are all non-negative and, when ordered in descending order, satisfy 
\begin{subequations}
\label{borlandanddennis}
\begin{eqnarray}
\lambda_1 + \lambda_6  = \lambda_2 + \lambda_5 =\lambda_3 + \lambda_4 =1, \quad
\label{equality constraints} \\
 \lambda_4 \ \le\  \lambda_5 + \lambda_6. \quad \quad \quad\quad  \quad
\label{inequality constraint}
\end{eqnarray}
\end{subequations}
\end{Corollary}
The equality constraints (\ref{equality constraints}) follow from the fact that each configuration in (\ref{eight}) contains either $h_1$ or $h_6$, but never both, and similarly for the other two pairs.  
The inequality constraint (\ref{inequality constraint}) follows from orthogonality relations between the coefficients in 
(\ref{eight}) that hold because the reference orbitals are natural orbitals.  We remark that, thanks to the $3$-qubit correspondence mentioned in the introduction, the inequality constraint also follows from the more general ``polygon inequalities'' that characterize the reduced density matrices of pure $n$-qubit states \cite{HiguchiSudberySzulc}.

A general solution of the $N$-representability problem for the 1RDM has since been achieved by Klyachko \cite{Klyachko2004,Klyachko2006}.  
For any number $N$ of fermions and any rank $R \ge N$, the spectra (ordered eigenvalue lists) of 1RDMs derived from $N$-fermion wave functions of rank $R$ or less form a convex polytope determined by finitely many inequality constraints.  These inequalities are called ``generalized Pauli constraints'' \cite{SchillingGrossChristandl} because they are stricter than the obvious constraints of the Pauli Exclusion Principle, namely, that orbital occupation numbers, and therefore in particular the ``natural occupation numbers,'' which are the eigenvalues of the 1RDM, are all less than or equal to $1$.   The number of generalized Pauli constraints (GPCs) grows very quickly with $N$ and $R$, however, and full sets of constraints are known explicitly only for fairly small $N$ and $R$ \cite{Klyachko2006}.   Fortunately, the GPCs for the Borland-Dennis setting ($N=3$ and $R=6$) are not so numerous and indeed the polytope of spectra can be visualized in three dimensions \cite{WalterDoranGrossChristandl,ChakrabortyMazziotti-Review}.

The GPCs are potentially of great significance for the study of electron structure, by virtue of what Klyachko has called the ``selection rule'' for ``pinned states'' \cite{KlyachkoPinning}.  
 When a many-fermion wave function ``saturates'' a GPC, that is, when it attains equality in the GPC inequality constraint, 
 then all the configurations in its natural configuration interaction (CI) expansion conform individually to the saturated constraint.  
The natural CI expansion of a pinned states is reduced by a sort of ``selection rule'' that eliminates configurations of natural orbitals that do not conform to the saturated GPC.   
A complete proof of the selection rule seems to have been published only very recently \cite{Walter Dissertation}; meanwhile, theorists have already begun to investigate its potential application to electronic structure theory  \cite{SchillingGrossChristandl,Lithium,Molecules,ChakrabortyMazziotti-Review,Schilling}.

Klyachko's selection rule implies the structure of the natural CI expansion for wave function in the Borland-Dennis setting.  
These wave functions saturate the GPCs $\lambda_i +  \lambda_{7-i} \le 1$.  
The selection rule implies that every configuration in the natural CI expansion of the wave function contains exactly one of the natural orbitals 
$h_i$ and $h_{7-i}$ (at least in the generic case where all six natural occupation numbers are distinct).  
This is precisely the structure of the natural CI expansion discovered by Borland and Dennis.

\subsection{3-qubit subspaces}

With the following definition, 
the Borland-Dennis Theorem may be stated quite succinctly:
{\it Every wave function $\psi \in \Borland$ belongs to a $3$-qubit subspace.}

\begin{Definition}
\label{3-qubit subspace definition}
A subspace of $\Borland$ is a ``$3$-qubit subspace'' if it has the form $\QQ_1\wedge \QQ_2 \wedge \QQ_3$ where $\QQ_1,\QQ_2$, and $\QQ_3$ are three mutually orthogonal $2$-dimensional subspaces 
 of $\Csix$. 
\end{Definition}

Let $\{g_1,\ldots,g_6\}$ be an orthonormal basis of $\Csix$ such that $\QQ_i = \Span\{g_i,g_{7-i}\}$, $i=1,2,3$.
A wave function belongs to the $3$-qubit subspace $ \QQ_1\wedge \QQ_2 \wedge \QQ_3$ 
if and only if its CI expansion with respect to the reference orbitals $g_i$ only involves configurations that have exactly one orbital in each pair $\{g_i,g_{7-i}\}$, $i=1,2,3$.  
The $3$-qubit subspaces are the so-called ``single-occupancy vector'' subspaces in the Borland-Dennis setting \cite{ChenChenDokovicZeng,ChenDokovicGrasslZeng}.

Theorem~\ref{Borland-Dennis Theorem} tells us that any wave function $\psi \in \Borland$ belongs to the $3$-qubit subspace 
\[ 
      \Span\{h_1,h_6\} \wedge \Span\{h_2,h_5\} \wedge \Span\{h_3,h_6\} \ ,
\]
where $h_1,\ldots,h_6$ are certain natural orbitals of $\psi$.
In fact, the theorem is equivalent to the assertion that every wave function $\psi \in \Borland$ belongs to {\it some} $3$-qubit subspace, 
not necessarily one involving natural orbitals.  Sufficiency of the latter condition is implied by the next lemma:

\begin{Lemma}
\label{3-qubit subspace lemma}
Suppose the wave function $\psi \in \Borland$ belongs to a $3$-qubit subspace $\QQ_1\wedge \QQ_2 \wedge \QQ_3$.  
Then there exists an orthonormal system $\{h_1,\ldots,h_6\}$ of natural orbitals such that $\QQ_i = \Span\{h_i,h_{7-i}\}$, $i=1,2,3$.  
\end{Lemma}

\begin{proof} 
Let $\{g_1,\ldots,g_6\}$ be an orthonormal basis of $\Csix$ such that $\QQ_i = \Span\{g_i,g_{7-i}\}$, $i=1,2,3$.
Then the CI expansion of $\psi$ with respect to the orbitals $f_i$ only involves configurations that have exactly one orbital in each pair $\{g_i,g_{7-i}\}$. 
 Because of this, the 1RDM of $\psi$, with respect to the ordered basis $(g_1,g_6,g_2,g_5,g_3,g_4)$, is block-diagonal with three $2\times 2$ blocks.  
 By diagonalizing these blocks one obtains an orthonormal system $\{h_1,\ldots,h_6\}$ of natural orbitals such that 
$\Span\{ h_i,h_{7-i} \} = \Span\{ g_i,g_{7-i} \}$.
\end{proof}

Thus, every wave function in the Borland-Dennis space belongs to at least one $3$-qubit subspace. 
Furthermore, Theorem~\ref{Borland-Dennis Theorem} implies that a {\it generic} wave function in the Borland-Dennis space belongs to a {\it unique} $3$-qubit subspace:  

\begin{Corollary}
\label{$3$-qubit subspace corollary}
A wave function $\psi \in \Borland$ that has distinct natural occupation numbers $\lambda_1 > \cdots > \lambda_6$ belongs to a unique $3$-qubit subspace.
\end{Corollary}

\begin{proof} 
Let $u_1,\ldots,u_6$ denote natural orbitals of $\psi$ with natural occupation numbers $\lambda_1,\ldots,\lambda_6$.  
The natural orbitals $u_i$ are unique up to phase by the hypothesis that the natural occupation numbers of $\psi$ are distinct. 
Suppose that $\psi$ belongs to a $3$-qubit subspace $\QQ_1\wedge \QQ_2 \wedge \QQ_3$.  We will prove that 
$\QQ_1\wedge \QQ_2 \wedge \QQ_3$ has to be $ \Span\{ u_1,u_6 \} \wedge \Span\{ u_2,u_5 \} \wedge \Span\{ u_3,u_4 \} $.

Let $\{g_1,\ldots,g_6\}$ be an orthonormal basis of $\Csix$ such that $\QQ_i = \Span\{g_i,g_{7-i}\}$, $i=1,2,3$.   
With respect to the ordered basis $(g_1,g_6,g_2,g_5,g_3,g_4)$, 
the 1RDM of $\psi$ is block-diagonal with three $2\times 2$ blocks.  
As in the proof of Lemma~\ref{3-qubit subspace lemma}, let $\{h_1,\ldots,h_6\}$ be an orthonormal system of natural orbitals such that 
$\Span\{ h_i,h_{7-i} \} = \Span\{ g_i,g_{7-i} \}$.   Let $\mu_i$ denote the natural occupation number corresponding to $h_i$.  
The pair of natural orbitals that belong to $\QQ_i$ is determined by the corresponding $2\times 2$ block of the 1RDM, for the trace of that block is $1$ and its determinant takes one of the three distinct values $\lambda_j\lambda_{7-j}$, $j=1,2,3$.   Thus $\{ \mu_i, \mu_{7-i} \}$ must be one of the pairs $\{ \lambda_j, \lambda_{7-j} \}$, and $\Span\{ h_i,h_{7-i} \} = \Span\{ u_j,u_{7-j} \} $.  
This implies that $\QQ_1\wedge \QQ_2 \wedge \QQ_3 = \wedge_{j=1}^3 \Span\{ u_j,u_{7-j} \}$.
\end{proof}

\subsection{Proofs of the Borland-Dennis Theorem}
\label{Proofs of the Borland-Dennis Theorem}

Borland and Dennis discovered Theorem~\ref{Borland-Dennis Theorem} through numerical experiments.  
M.B. Ruskai and R.L. Kingsley reportedly proved it not long afterwards \cite{borland2}, but their proofs were not published at the time.  
Ruskai's proof was finally published in 2007 \cite{Ruskai}.  
The theorem has been obtained as by-product of deeper theories: we have already mentioned that (1) it is a consequence of Klyachko's selection rule, and (2) it is the first special case of the more general theorem of Ref.~\cite{ChenChenDokovicZeng} that all $3$-fermion wave functions are ``single-occupancy vectors.'' 

We will present two elementary proofs of the Borland-Dennis Theorem, one here and the other in the appendix.   
The SVD Lemma is essential for both proofs.   Both proofs arrive at the Borland-Dennis Theorem by building a $3$-qubit subspace containing an arbitrary wave function.  
The proof given just below does this very elegantly, but in a nonconstructive way.  
The proof in the appendix resembles Ruskai's proof and is constructive: it tells you how to build a $3$-qubit subspace starting from any natural orbital.

To prove the Borland-Dennis Theorem, it suffices, by Lemma~\ref{3-qubit subspace lemma}, to prove that every wave function in $\Borland$ belongs to some $3$-qubit subspace.   
The following proposition exhibits such a $3$-qubit subspace.  

\begin{Proposition}
\label{Proposition for CID}
Let $\chi_*$ be a max-overlap CID approximation of a wave function $\psi \in \Borland$ with CI expansion 
(\ref{canonical form CID}).  
Then $\psi$ can be written as a linear combination of eight configurations in the max-overlap CID orbitals, namely, the four 
configurations in (\ref{canonical form CID}) and their duals, obtained by swapping X's and O's. 
\end{Proposition}
\begin{proof}

Thanks to the symmetry of the form (\ref{canonical form CID}), any of its four configurations could serve as the reference configuration.  
That is,  $\chi_*$ belongs to the CID subspace
(\ref{CID subspace}) with reference space $\RR_{123} = \hbox{span}\{ f_1 , f_2 , f_3 \}$; but it also belongs to the CID subpsace 
with reference space $\RR_{156} = \hbox{span}\{ f_1,  f_5 , f_6 \}$, as well as the CID subspaces $\Ffrak_{246}$ and $\Ffrak_{345}$, similarly defined.  

Because $\chi_*$ is a max-overlap CID approximation of $\psi$, it must be proportional to the projections of $\psi$ onto each of the subspaces 
$\Ffrak_{123} , \Ffrak_{156}, \Ffrak _{246},$ and $\Ffrak _{345}$.  Therefore, the CI expansion of $\psi$ with respect to the orbitals $f_1, \ldots,f_6$ cannot contain any double excitations of $f_1\wedge f_2 \wedge f_3$ or $f_1\wedge f_5\wedge f_6$ or $f_2\wedge f_4 \wedge f_6$ or $f_3\wedge f_4 \wedge f_5$ (besides the ones already in this list).  For example,  since $f_1\wedge f_4 \wedge f_5$ is a double excitation from the reference $f_1\wedge f_2 \wedge f_3$ that does not already appear in (\ref{canonical form CID}), its coefficient in the CI expansion of $\psi$ must equal $0$.  Striking off all such configurations from the list of all $20$ possible configurations leaves only the eight described in the statement of the theorem.  
\end{proof}

Proposition~\ref{Proposition for CID} implies that any $\psi \in \Borland$ belongs to the $3$-qubit subspace 
\[
    \Span\{f_1,f_4\}\wedge \Span\{f_2,f_5\} \wedge \Span\{f_3,f_6\}
\]
where the orbitals $f_i$ are the reference orbitals for the canonical form (\ref{canonical form CID}) of $\chi_*$. 
Theorem~\ref{Borland-Dennis Theorem} follows by Lemma~\ref{3-qubit subspace lemma}.


\section{The 5-term canonical forms}
\label{The 5-term canonical forms}


\subsection{The symmetric five-term canonical form}
\label{The symmetric five-term canonical form}

Recall that orthonormal orbitals $f_1, f_2,f_3$ are called ``best-overlap orbitals'' \cite{KutzelniggSmith} for a wave function $\psi \in \Borland$ if $f_1 \wedge f_2 \wedge  f_3$ is a 
max-overlap Slater determinant approximation of $\psi$. 
 In a CI expansion with respect to a reference basis that includes best-overlap orbitals, configurations that are singly excited from the reference configuration do not appear, cf., Lemma~\ref{Brueckner lemma}.   Thanks to these ``Brueckner conditions'' \cite{Nesbet,Larsson,Ortiz} and a bit of algebra, one can show that any wave function in the Borland-Dennis space has a symmetric $5$-term CI expansion.

\begin{Theorem}
\label{Theorem for Slater}
Let $\psi$ be a wave function in the Borland-Dennis space and let $\chi_*$ be a Slater determinant such that 
\[
      \langle \chi_*,\psi \rangle \ \ge \ \big|\langle \chi, \psi  \rangle\big|
\]
for all other Slater determinants $\chi$.  
Then there exists an orthonormal basis $ \{f_1, f_2, f_3 ,g_1, g_2, g_3 \}$ of $\Csix$ such that 
$\chi_* = g_1 \wedge  g_2 \wedge  g_3$ and 
\[
   \psi   \equals     \langle  \chi_* , \psi  \rangle \  g_1 \wedge  g_2 \wedge  g_3 \plus   A\  f_1 \wedge f_2 \wedge f_3  \plus  B_1\  g_1 \wedge f_2 \wedge f_3  \plus  B_2 \  f_1 \wedge g_2 \wedge f_3   \plus  B_3 \  f_1 \wedge f_2 \wedge g_3\ .
\]
\end{Theorem}

\begin{proof}
Let $g'_1, g'_2,$ and $g'_3$ be orthonormal orbitals such that $\chi_* = g'_1 \wedge  g'_2 \wedge  g'_3$.   
Let  $\RR$ denote the orthogonal complement of  $\RR^\perp = \Span\{g'_1, g'_2, g'_3 \}$, 
and let $\{ f'_1, f'_2, f'_3 \} $ be an orthonormal basis of $\RR$.
By Lemma~\ref{Brueckner lemma}, the CI expansion of $\psi$ with respect to the ordered basis $(f'_1, f'_2, f'_3 ,g'_1, g'_2, g'_3)$
lacks all nine configurations $f'_i \wedge  g'_j \wedge  g'_k $ that are singly-excited from $g'_1 \wedge  g'_2 \wedge  g'_3$.    
This means that $\psi$ belongs to the subspace (\ref{singles-free subspace}) of $\Borland$, and can therefore be written as in (\ref{5-term expansion}). 
The phases of the orbitals $g_i$ can be chosen {\it ad libitum} without changing the form of the CI expansion; let them be adjusted so that $g_1 \wedge  g_2 \wedge  g_3$ equals $\chi_*$.  
The coefficient $D$ in (\ref{5-term expansion}) is then equal to $\langle \chi_*, \psi \rangle$.  
\end{proof}

Theorem~\ref{Theorem for Slater} implies that every Borland-Dennis wave function $\psi$ is U(6)-equivalent to a wave function 
\begin{equation} 
\label{Chen et al. canonical form}  
a\ e_2 \wedge  e_3 \wedge  e_5 \plus   b\  e_1 \wedge  e_4 \wedge  e_5  \plus  c\ e_1 \wedge  e_3 \wedge  e_6  \plus
d \ e_2 \wedge  e_4 \wedge  e_6  \plus    z \ e_1 \wedge  e_3 \wedge  e_5 
\end{equation}
such that $a \ge b \ge c \ge 0$, $\mathrm{Re}(z) \ge 0$, and $d$ is equal to the maximum overlap of $\psi$ with a Slater determinant.
Chen {\it et al.} \cite{ChenDokovicGrasslZeng} show that the standardized wave function in (\ref{Chen et al. canonical form}) can serve as a canonical form for wave functions in the Borland-Dennis setting.  
It is a ``canonical'' form in the sense that wave functions with different canonical forms are not equivalent, except perhaps when their representatives are on the boundary of the ``canonical region.''   The canonical region is demarcated by a few polynomial inequalities in the coefficients $a,b,c,d,$ and $z$.

\subsection{Lone-orbital expansions}
\label{Lone-orbital expansions}

In this section we construct the analog in the Borland-Dennis setting of the ``generalized Schmidt decomposition'' for $3$-qubits \cite{AcinAndrianovJaneTarrach 2000, AcinAndrianovJaneTarrach 2001}.  
We call it the ``lone-orbital'' CI expansion because it features an solitary orbital that deigns to participate in at most one configuration.

\begin{Theorem}
\label{lone orbital theorem}
Let $w_1$ and $w_2$ be two orthonormal orbitals in $\Csix$ and let $\MM$ denote the orthogonal complement of $\Span\{ w_1,w_2 \}$ in $\Csix$.   
Suppose 
\begin{equation}
\label{partitioned}
\psi \equals  w_1 \wedge  \gamma_1 \plus  w_2 \wedge \gamma_2
\end{equation}
is a wave function in $\Borland$ such that $\gamma_1, \gamma_2 \in  \wedge^2 \MM$.
Then there exists an orthonormal system of reference orbitals $( l_1, l_2, f_1,f_2,g_1,g_2)$ with respect to which $\psi$ has CI expansion 
\begin{equation}
\label{lone orbital expansion}
\begin{matrix}
 l_1 & l_2  & f_1 & f_2 & g_1 & g_2 & \\
\mathrm{X} &  \mathrm{O} & \mathrm{X} & \mathrm{X} & \mathrm{O} & \mathrm{O} & \qquad A   \\
\mathrm{O} & \mathrm{X} & \mathrm{X} & \mathrm{X} & \mathrm{O} & \mathrm{O} &  \qquad B_1  \\
\mathrm{O} & \mathrm{X} & \mathrm{O} & \mathrm{O} & \mathrm{X} & \mathrm{X} & \qquad B_2 \\
\mathrm{O} & \mathrm{X} & \mathrm{X} & \mathrm{O} & \mathrm{X} & \mathrm{O} & \qquad C_1   \\
\mathrm{O} & \mathrm{X} & \mathrm{O} & \mathrm{X} & \mathrm{O} & \mathrm{X} &  \qquad C_2 \\
\end{matrix}
\end{equation}
and such that $ \Span\{ l_1, l_2 \} = \Span\{ w_1,w_2\}$.
\end{Theorem}

\begin{proof}
First suppose $\gamma_1$ and $ \gamma_2 $ are linearly independent.  Then, by Lemma \ref{decomposable in span}, $\Span\{\gamma_1, \gamma_2\}$ contains a Slater determinant geminal $ f'_1\wedge  f'_2$, where $ f'_1$ and $ f'_2$ are orthonormal orbitals in $\WW$.  

Let $A_1, A_2$ be scalars such that $|A_1|^2+|A_2|^2=1$ and $A_1\gamma_1 + A_2 \gamma_2 =  A f'_1\wedge  f'_2$ with $A = \| A_1\gamma_1 + A_2 \gamma_2 \|$.  
Define orbitals $l_1$ and $l_2$ by the unitary transformation
\[
   \begin{bmatrix} w_1 \\ w_2 \end{bmatrix}
    \equals 
   \begin{bmatrix}  A_1 &  - \overline{A}_2 \\    A_2 & \ \overline{A}_1   \end{bmatrix} 
   \begin{bmatrix} l_1 \\ l_2  \end{bmatrix} 
\]
Then 
\begin{eqnarray*}
\psi \equals (A_1 l_1 - \overline{A}_2 l_2 ) \wedge  \gamma_1  +  (  A_2 l_1 + \overline{A}_1 l_2) \wedge  \gamma_2  \nonumber 
             & \equals  &  l_1  \wedge   ( A_1 \gamma_1 +  A_2 \gamma_2 ) +   l_2 \wedge   ( \overline{A}_1 \gamma_2-  \overline{A}_2 \gamma_1  )   \nonumber \\
              & \equals & A \ l_1  \wedge   f'_1 \wedge  f'_2 \plus   l_2 \wedge   \gamma
\end{eqnarray*}
where $\gamma = \overline{A}_1 \gamma_2-  \overline{A}_2 \gamma_1 $.

The CI expansion of $\psi$ with respect to the ordered basis $(l_1, l_2, f'_1,f'_2,g'_1,g'_2)$ is then  
\begin{equation}
\label{ready for SVD}
\begin{matrix}
 l_1 & l_2 & f'_1 & f'_2 & g'_1 & g'_2 & \quad &  \\
  \mathrm{X} &  \mathrm{O} &  \mathrm{X} & \mathrm{X} & \mathrm{O} & \mathrm{O} & &  A   \\ 
  \mathrm{O} &  \mathrm{X} & \mathrm{X} & \mathrm{X} & \mathrm{O} & \mathrm{O}  & &  B_1  \\
  \mathrm{O} &  \mathrm{X} & \mathrm{O} & \mathrm{O} & \mathrm{X} & \mathrm{X}  & &  B_2     \\ 
 \mathrm{O}  &  \mathrm{X} & \mathrm{X} & \mathrm{O} & \mathrm{X} & \mathrm{O} & &   M_{11}  \\
 \mathrm{O}  &  \mathrm{X} & \mathrm{X} & \mathrm{O} & \mathrm{O} & \mathrm{X}  & &   M_{12}  \\ 
 \mathrm{O} &  \mathrm{X} & \mathrm{O} & \mathrm{X} & \mathrm{X} & \mathrm{O}  & &   M_{21}  \\ 
 \mathrm{O} &  \mathrm{X}  & \mathrm{O} & \mathrm{X} & \mathrm{O} & \mathrm{X} & &   M_{22}  \\
\end{matrix}
\end{equation}
The sum of the last four lines of (\ref{ready for SVD}) is $l_2 \wedge \sum\limits_{i,j=1}^2 M_{ij} \  f'_i\wedge  g'_j$.
This can be reduced to a sum of two terms by using the singular value decomposition of the matrix ${\bf  M} = (M_{ij})$.
 Let ${\bf U}= (U_{ij})$ and ${\bf V}= (V_{ij}) $ be unitary matrices such that $ {\bf U}^* {\bf M} {\bf V} $ is the $2\times 2$ diagonal matrix with diagonal entries $C_1$ and $C_2$.  
We may assume that $\det({\bf U})$ and $\det({\bf V}) $ are both equal to $1$ because we don't require $C_1$ and $C_2$ to be positive.
Define new orbitals $f_1,f_2,g_1,$ and $g_2$ by the unitary transformations
\begin{equation}
 \label{these substitutions}
  f_i = \sum_{k=1,2}  \overline{U}_{ik} f'_k
          \qquad \hbox{and} \qquad 
  g_j = \sum_{k=1,2}  V_{jk} g'_k
\end{equation}
Substituting (\ref{these substitutions}) into (\ref{ready for SVD}) transforms the sum of its last four terms into 
$ C_1 \ l_2 \wedge f_1  \wedge  g_1 + C_2 \ l_2 \wedge f_2  \wedge  g_2 $  
without changing the coefficients $A,B_1$, and $B_2$.  
The CI expansion (\ref{lone orbital expansion}) results.  
\end{proof}

Any wave function in $\Borland$ can be written as in (\ref{partitioned}) by Proposition~\ref{max-overlap rank-5 proposition}, or by the Borland-Dennis Theorem, and therefore:
\begin{Corollary}
\label{lone orbital existence}
Every wave function in $\Borland$ has CI expansions of the form (\ref{lone orbital expansion}).
\end{Corollary}

\begin{Corollary}
\label{all natural}
All six natural occupation numbers of a wave function $\psi \in \wedge^3 \mathbb{C}^6$ equal $1/2$ 
if and only $\psi$ has a CI expansion of the form 
$\tfrac{1}{\sqrt{2}} \big( u_1\wedge  u_2 \wedge  u_3 \plus   u_4\wedge  u_5 \wedge  u_6 \big)$.
\end{Corollary}

\begin{proof}
The sufficiency of the condition is clear; we prove its necessity.

Let (\ref{lone orbital expansion}) be a lone-orbital CI expansion of $\psi$.  
Since all six natural occupation numbers of $\psi$ are equal, the 1RDM of $\psi$ is proportional to the $6 \times 6$ identity matrix, and therefore {\it every} orbital in $\mathbb{C}^6$ is a natural orbital with occupation number $1/2$.  Since $l_1$ has occupation number $1/2$, it must be that $|A|^2=1/2$.  But then, since $f_1$ and $f_2$ also have occupation number $1/2$, they cannot appear in any other configurations.   That leaves two configurations in the lone-orbital CI expansion.  Renaming the reference orbitals and adjusting their phases yields the stated form for $\psi$.
\end{proof}

\begin{Corollary}
\label{hyperdet = 0 iff CIS}
A wave function $\psi \in \Borland$ is CIS if and only if its hyperdeterminant equals $0$.  
\end{Corollary}

\begin{proof}
The hyperdeterminant of a CIS wave function in its canonical form (\ref{canonical form CIS}) can be directly computed from its coefficients (cf., eg., formula (32) in Ref.~\cite{ChenDokovicGrasslZeng}) and it equals $0$.  
Conversely, let $\psi$ be a wave function whose hyperdeterminant equals $0$, and let (\ref{lone orbital expansion}) be any lone orbital CI expansion of $\psi$.  
In terms of the coefficients in (\ref{lone orbital expansion}), the hyperdeterminant is simply $A^2B_2^2$.  
This equals $0$ if and only if $A = 0$ or $B_2 = 0$.   If $B_2 = 0$, then $\psi$ is manifestly of CIS form with reference space $\RR = \Span\{l_2,f_1,f_2\}$.   
If $A = 0$, then $\psi$ is low-rank and is CIS {\it a fortiori}.
\end{proof}

\section{5-term expansions based on max-overlap CIS approximations}
\label{5-term expansions based on max-overlap CIS approximations}

Theorem~\ref{Theorem for Slater} states that, given any max-overlap Slater determinant approximation $\chi_*$ of a target wave function, the target has a $5$-term CI expansion whose leading configuration is $\chi_*$.
In this section we prove a similar theorem for max-overlap CIS  approximations:  
if $\chi_*$ is a max-overlap CIS approximation of a target wave function $\psi$, then there exists $D \in \mathbb{C}$ such that
\[ \psi \equals  \langle \chi_*, \psi \rangle \ \chi_* \plus D\ g_1 \wedge g_2 \wedge g_3\ , \]
where the orbitals $g_i$ are the ones in the canonical form (\ref{canonical form CIS}) of $\chi_*$.  

\begin{Theorem}
\label{Theorem for CIS}
Let $\chi_* \in \Borland$ be a CIS wave function with CI expansion (\ref{canonical form CIS}).  
If $\chi_*$ is a max-overlap CIS approximation of a wave function $\psi \in \Borland$, then $\psi$ has CI expansion 
\begin{equation}
\label{max-overlap CIS expansion}
\begin{matrix}
 f_1 & f_2 & f_3 & g_1 & g_2 & g_3 & \\
\mathrm{X} & \mathrm{X} & \mathrm{X} & \mathrm{O} & \mathrm{O} & \mathrm{O} & \qquad sA  \\
\mathrm{O} & \mathrm{X} & \mathrm{X} & \mathrm{X} & \mathrm{O} & \mathrm{O} & \qquad sB_1  \\
\mathrm{X} & \mathrm{O} & \mathrm{X} & \mathrm{O} & \mathrm{X} & \mathrm{O} & \quad  -sB_2  \\
\mathrm{X} & \mathrm{X} & \mathrm{O} & \mathrm{O} & \mathrm{O} & \mathrm{X} & \qquad sB_3  \\
\mathrm{O} & \mathrm{O} & \mathrm{O} & \mathrm{X} & \mathrm{X} & \mathrm{X} & \qquad D   \\
\end{matrix}
\end{equation}
where $s  =  \langle \chi_*, \psi \rangle $.  
If $\psi$ is of full rank then so is $\chi_*$.  
\end{Theorem}

The $5$-term CI expansion (\ref{max-overlap CIS expansion}) has the same shape as the ``canonical form'' based on max-overlap Slater determinant approximations.  Accordingly, many of the techniques developed in Ref.~\cite{ChenDokovicGrasslZeng} are also useful for studying max-overlap CIS approximations.    In particular, the algorithm described there for computing the maximum overlap with a Slater determinant can be adapted to the problem of computing the maximum overlap with a CIS wave function.  The adapted algorithm should also yield the coefficients $A,B_1,B_2,B_3$ of the canonical forms of all the max-overlap CIS approximations.  

\subsection{Proof of the main theorem}
\label{Proof of the main theorem}

This section is dedicated to the proof of Theorem~\ref{Theorem for CIS}.    
The expansion (\ref{max-overlap CIS expansion}) will be established first for the case where $\chi_*$ is of full rank; this is Lemma~\ref{Lemma for CIS 2}.  
We will then prove that, when $\chi_*$ is not of full rank, then $\chi_*$ equals $\psi$.
This implies that $\psi$ is of full rank if and only if $\chi_*$ is.  

\begin{Lemma}
\label{Lemma for CIS 1}
Let $\chi_* \in \Borland$ be a CIS wave function with CI expansion (\ref{canonical form CIS}).  
If $\chi_*$ is a max-overlap CIS approximation of a wave function $\psi \in \Borland$, then $\psi$ has CI expansion 
\begin{equation}
\label{8-term config diagram for lemma}
\begin{matrix}
 f_1 & f_2 & f_3 & g_1 & g_2 & g_3 & \\
\mathrm{X} & \mathrm{X} & \mathrm{X} & \mathrm{O} & \mathrm{O} & \mathrm{O} & \qquad sA  \\
\mathrm{O} & \mathrm{X} & \mathrm{X} & \mathrm{X} & \mathrm{O} & \mathrm{O} & \qquad sB_1  \\
\mathrm{X} & \mathrm{O} & \mathrm{X} & \mathrm{O} & \mathrm{X} & \mathrm{O} & \quad  -sB_2  \\
\mathrm{X} & \mathrm{X} & \mathrm{O} & \mathrm{O} & \mathrm{O} & \mathrm{X} & \qquad sB_3  \\
\mathrm{O} & \mathrm{O} & \mathrm{X} & \mathrm{X} & \mathrm{X} & \mathrm{O} & \qquad C_3  \\
\mathrm{O} & \mathrm{X} & \mathrm{O} & \mathrm{X} & \mathrm{O} & \mathrm{X} & \quad  -C_2  \\
\mathrm{X} & \mathrm{O} & \mathrm{O} & \mathrm{O} & \mathrm{X} & \mathrm{X} & \qquad C_1  \\
 \mathrm{O} & \mathrm{O} & \mathrm{O} & \mathrm{X} & \mathrm{X} & \mathrm{X} & \qquad D   \\
\end{matrix}
\end{equation}
where $s  =  \langle \chi_*, \psi \rangle $.
\end{Lemma}

\begin{proof}
$\chi_*$ is CIS with reference space $\RR = \Span\{f_1,f_2,f_3\}$.  Since $\chi_*$ is a max-overlap CIS approximation of $\psi$, it must have maximum overlap with $\psi$ among all wave functions in the subspace $\wedge^3 \RR  \oplus ( \RR  \wedge \RR  \wedge \RR^\perp)$ and, by Lemma~\ref{useful little observation}, it must be proportional to the projection of $\psi$ onto this subspace.  It follows that the CI expansion of $\psi$ with respect to the reference orbitals $f_1, f_2, \ldots, g_3$ is given by the sum of the first four rows of the configuration diagram (\ref{8-term config diagram for lemma}).  The coefficients of all other single excitations of the reference configuration $f_1 \wedge  f_2 \wedge  f_3$ all equal $0$.  

Now consider the double excitations of the reference configuration.  There are nine of them.  Three of them appear in the CI expansion (\ref{8-term config diagram for lemma}) with coefficients $C_1,C_2$, and $C_3$; we have to prove that the coefficients of the other six configurations are $0$.  
The absent configurations are those of the form $f_i\wedge g_j \wedge g_k$ with $ i \in \{j,k\} \subset \{ 1,2,3 \}$ and $j < k$.  

For the sake of concreteness but without loss of generality, we will work on $f_1\wedge g_1 \wedge g_2$.  
The five-dimensional subspace of $\Borland$ spanned by $f_1\wedge g_1 \wedge g_2$ together with the four configurations in (\ref{canonical form CIS}) consists entirely of CIS wave functions, for the hyperdeterminant of any such wave function equals $0$ by formula (14) in Ref.~\cite{LevayVrana}.  Therefore the projection of $\psi$ onto this subspace is proportional to $\chi_*$, and the coefficient of $f_1\wedge g_1 \wedge g_2$ in the CI expansion of $\psi$ must equal $0$.  Similar arguments show that the coefficients of the other five double excitations of this kind also vanish, and therefore the CI expansion of $\psi$ is as in (\ref{8-term config diagram for lemma}).
\end{proof}

\begin{Lemma}
\label{Lemma for CIS 2}
Let $\chi_*$ be a max-overlap CIS approximation of a wave function $\psi \in \Borland$.  
Suppose that  $\chi_*$ is of full rank and has canonical CI expansion (\ref{canonical form CIS}).     
Then the coefficients $C_1,C_2$, and $C_3$ in the CI expansion (\ref{8-term config diagram for lemma}) of $\psi$ all equal $0$.
\end{Lemma}

\begin{proof}

Since $\chi_*$ is of full rank, the coefficients $B_1,B_2$, and $B_3$ in its canonical CI expansion (\ref{canonical form CIS}) are all positive.  
Let  $\varphi_1$ and $\varphi_2$ denote $f_1\wedge g_2 \wedge g_3$ and $g_1\wedge f_2 \wedge g_3$, respectively.  
Let 
\begin{eqnarray*}
  \varphi'_1 & = &  \frac{B_1C_2+B_2C_1}{B_1^2+B_2^2}(B_2 \varphi_1 + B_1 \varphi_2) \\
  \varphi'_2 & = &  \frac{B_1C_1-B_2C_2}{B_1^2+B_2^2}(B_1 \varphi_1- B_2 \varphi_2).
\end{eqnarray*}
The trivectors $\varphi'_1$ and $\varphi'_2$ are orthogonal and 
$ C_1 \varphi_1+ C_2 \varphi_2  = \varphi'_1 + \varphi'_2 $.
Referring to the configuration diagram (\ref{8-term config diagram for lemma}) for $\psi$, one can see that 
\[
   \psi \equals  ( s \chi_*  +  \varphi'_1) \plus (\varphi'_2 \plus C_3  \ g_1\wedge g_2 \wedge f_3 \plus D \  g_1\wedge g_2 \wedge g_3)\ .
\]
The two grouped expressions on the right-hand side of this equation represent orthogonal trivectors.  
The first grouped term, namely, $s \chi_*  \plus  \varphi'_1$ has hyperdeterminant $0$, as can be seen by formula (37) in Ref.~\cite{LevayVrana}.  
It is therefore CIS.  Since  $\chi_*$ is a max-overlap CIS approximation of $\psi$ by assumption, and since $\varphi'_1$ is orthogonal to $\chi_*$, it follows that $\varphi'_1 = 0$, whence 
$    B_1C_2 = -B_2C_1  $.
Similar arguments show that $B_2C_3=-B_3C_2$ and $B_3C_1=-B_1C_3$.  As we are assuming that $B_1B_2B_3 \ne 0$, the three preceding equations imply that $C_1,C_2$, and $C_3$ are all $0$. 
\end{proof}

We now complete the proof of Theorem~\ref{Theorem for CIS}.  
The conclusion has already been established in Lemma~\ref{Lemma for CIS 2} for the case where $\chi_*$ is of full rank.
The rest of this proof is devoted to the case where $\chi_*$ is low-rank.  
We will show in these cases that $\chi_* = \psi$, so that $\psi$ has the CI expansion (\ref{max-overlap CIS expansion}) with $s = 1,D=0$.  The fact that $\chi_* = \psi$ if $\chi_*$ is low-rank implies the final assertion of the theorem.  

Suppose $\chi_*$ has rank $3$, i.e., that $\chi_*$ is a Slater determinant.  Then $\chi_*$ can be written as in (\ref{canonical form CIS}) with $A=1$ and $B_i = 0$, $i=1,2,3$.  By Lemma~\ref{Lemma for CIS 1}, the target wave function $\psi$ has CI expansion (\ref{8-term config diagram for lemma}).  Since $\chi_*$ is a max-overlap CIS approximation of $\psi$, it is also a max-overlap low-rank approximation. 
 This implies that each of the coefficients $C_i$ in (\ref{8-term config diagram for lemma}) must equal $0$, and only the coefficients $sA = \langle \chi_*, \psi \rangle$ and $D$ remain.   
Without changing notation, let the phase of $g_1$ be adjusted so that 
\begin{equation}         
\label{ortho-GHZ wave function}
\psi \equals  S \  f_1 \wedge f_2 \wedge f_3 \ +\  T \  g_1 \wedge g_2 \wedge g_3
\end{equation} 
where  
$S =  \langle \chi_*, \psi \rangle > 0 $ and $T = |D| \ge 0$.  
The tandem rotations 
\begin{equation}
\label{tandem}
   \begin{bmatrix} f_i(t) \\ g_i (t) \end{bmatrix}
    \equals 
   \begin{bmatrix}  T &  -S \\    S & \ \ T   \end{bmatrix} 
   \begin{bmatrix} f'_i \\ g'_i  \end{bmatrix} 
\end{equation} 
$i = 1,2,3$ convert (\ref{ortho-GHZ wave function}) to 
\begin{equation}
\label{max-overlap CIS expansion of ortho-GHZ}
\begin{matrix}
 f'_1 & f'_2 & f'_3 & g'_1 & g'_2 & g'_3 & \\
 \mathrm{X} & \mathrm{X} & \mathrm{X} & \mathrm{O} & \mathrm{O} & \mathrm{O} & \qquad ST  \\
\mathrm{X} & \mathrm{O} & \mathrm{O} & \mathrm{O} & \mathrm{X} & \mathrm{X} & \qquad ST  \\
\mathrm{O} & \mathrm{X} & \mathrm{O} & \mathrm{X} & \mathrm{O} & \mathrm{X} & \qquad   ST  \\
\mathrm{O} & \mathrm{O} & \mathrm{X} & \mathrm{X} & \mathrm{X} & \mathrm{O} & \qquad  ST  \\
\mathrm{O} & \mathrm{O} & \mathrm{O} & \mathrm{X} & \mathrm{X} & \mathrm{X} & \qquad\quad S^2 - T^2\ .   \\
\end{matrix}
\end{equation} 
The sum of the last four rows in (\ref{max-overlap CIS expansion of ortho-GHZ}), renormalized, is a CIS wave function whose overlap-squared with $\psi$ is $ 1 - (ST)^2 $.  
But $ 1 - (ST)^2 \ge S^2 =   \langle \chi_*, \psi \rangle ^2$ with equality only if $S = 1, T=0$.  As $\chi_*$ is assumed to be a max-overlap CIS approximation of $\psi$, necessarily  $S = 1$ and $T=0$, which means that $\psi$ is a Slater determinant and therefore $\chi_* = \psi$.

Next, suppose that $\chi_*$ has rank $5$.  
As mentioned in Section~\ref{Canonical forms}, most rank-$5$ wave functions can be written in the form (\ref{canonical form CIS}) with $A,B_1,B_2 > 0$ and $B_3 = 0$.  
The only exceptions are rank-$5$ wave functions that are equivalent to 
\begin{equation}
\label{annoying special case}
    \tfrac{1}{\sqrt{2}} e_1 \wedge \big(  e_2 \wedge e_3 + e_4 \wedge e_5\big).
\end{equation}

Suppose first that $\chi_*$ is as in (\ref{canonical form CIS}) with $A,B_1,B_2 > 0$ and $B_3 = 0$.  
By Lemma~\ref{Lemma for CIS 1}, the target wave function $\psi$ has CI expansion (\ref{8-term config diagram for lemma}).   The coefficient $C_3$ in (\ref{8-term config diagram for lemma}) must be $0$ or else $\psi$ would have a better low-rank approximation --- and hence a better CIS approximation --- than $\chi_*$.  
Arguing as in the proof of Lemma~\ref{Lemma for CIS 2}, one can show that $AC_1=-B_1D$, $B_1C_2=-B_2C_1$, and $B_2D=-AC_2$.  As $A,B_1,B_2 > 0$, these three equations imply that $C_1 = C_2 = D = 0$.  Thus, $\psi$ has CI expansion (\ref{max-overlap CIS expansion}) with 
$s=1$.  In other words, $\psi$ is itself of rank $5$, and therefore $\chi_* = \psi$.

Finally, suppose that $\chi_*$ is equivalent to (\ref{annoying special case}).  Then $\chi_*$ can be written as in (\ref{canonical form CIS}) and the target wave function $\psi$ can be written as in (\ref{8-term config diagram for lemma}) with $A = 0,B_3=0$, and $B_1 = B_2 =  \tfrac{1}{\sqrt{2}} $.    The coefficients $C_3$ and $D$ in (\ref{8-term config diagram for lemma}) must also equal $0$, otherwise $\psi$ would have a strictly better CIS approximation with reference space $\Span\{f_3,g_1,g_2\}$.  Moreover, since $\chi_*$ is a max-overlap CIS approximation of rank-5, it is also a max-overlap low-rank approximation, and therefore, by Proposition~\ref{max-overlap rank-5 proposition}, $f_3$ and $g_3$ are natural orbitals of $\psi$.   This implies that $\psi$ has CI expansion  
\begin{equation}
\label{4-term config diagram for stupid special case}
\begin{matrix}
 f_1 & f_2 & f_3 & g_1 & g_2 & g_3 & \\
\mathrm{O} & \mathrm{X} & \mathrm{X} & \mathrm{X} & \mathrm{O} & \mathrm{O} & \qquad B  \\
\mathrm{X} & \mathrm{O} & \mathrm{X} & \mathrm{O} & \mathrm{X} & \mathrm{O} & \quad  -B  \\
\mathrm{O} & \mathrm{X} & \mathrm{O} & \mathrm{X} & \mathrm{O} & \mathrm{X} & \qquad  C  \\
\mathrm{X} & \mathrm{O} & \mathrm{O} & \mathrm{O} & \mathrm{X} & \mathrm{X} & \qquad  C  \\
\end{matrix}
\end{equation}
where $B = \tfrac{1}{\sqrt{2}} \langle \chi_*, \psi \rangle$ and $C$ is such that $B^2 + |C|^2 = 1/2$. 
The tandem rotations (\ref{tandem}) for $i = 1,2$ with $S = T = 1/\sqrt{2}$
transform the CI expansion (\ref{4-term config diagram for stupid special case}) to 
\begin{equation}
\label{another 4-term config diagram for stupid special case}
\begin{matrix}
 f'_1 & f'_2 & f_3 & g'_1 & g'_2 & g_3 & \\
\mathrm{X} & \mathrm{X} & \mathrm{X} & \mathrm{O} & \mathrm{O} & \mathrm{O} & \qquad B  \\
\mathrm{O} & \mathrm{O} & \mathrm{X} & \mathrm{X} & \mathrm{X} & \mathrm{O} & \quad  -B  \\
\mathrm{O} & \mathrm{X} & \mathrm{O} & \mathrm{X} & \mathrm{O} & \mathrm{X} & \qquad  C  \\
\mathrm{X} & \mathrm{O} & \mathrm{O} & \mathrm{O} & \mathrm{X} & \mathrm{X} & \qquad  C  \\
\end{matrix}
\end{equation}
The sum of the first three rows of (\ref{another 4-term config diagram for stupid special case}) is proportional to a CIS wave function with reference space $\Span\{f'_2,f_3,g'_1\}$ whose overlap-squared with $\psi$ is $2B^2 + |C|^2$.  This is strictly greater than 
$  \langle \chi_*, \psi \rangle^2 = 2B^2$ unless $C=0$.  Since $\chi_*$ is a max-overlap CIS approximation, $C$ must equal $0$ in  (\ref{4-term config diagram for stupid special case}), which implies that $\psi$ has rank $5$ and therefore $\chi_* = \psi$.

\subsection{Analog of the CIS expansion in the 3-qubit setting}
\label{Analog of the CIS expansion in the 3-qubit setting}

In this section we prove the analog of Theorem~\ref{Theorem for CIS} for max-overlap ``Type~4a'' approximations in the $3$-qubit setting.   

A wave function in the $3$-qubit space $\mathbb{C}^2 \otimes \mathbb{C}^2 \otimes \mathbb{C}^2$ is a ``Type~4a state'' \cite{AcinAndrianovJaneTarrach 2000,AcinAndrianovJaneTarrach 2001} 
if it can be written in the form 
\[
    A \ f_1\otimes f_2 \otimes f_3 \plus B_1 \ g_1\otimes f_2 \otimes f_3 \plus B_2 \ f_1\otimes g_2 \otimes f_3 \plus B_3 \ f_1\otimes f_2 \otimes g_3 
\]
where $\{f_1,g_1\},\{f_2,g_2\}$, and $\{f_3,g_3\}$ are three orthonormal bases of $\mathbb{C}^2$.  
In other words, a $3$-qubit wave function is a Type~4a state if it is locally unitarily equivalent to 
$ A |000\rangle + B_1  |100\rangle + B_2  |010\rangle + B_3  |001\rangle$.   
We include all degenerate forms in Type~4a, allowing one or more of the coefficients to be $0$.  

A $3$-qubit wave function is a Type~4a state if and only if its hyperdeterminant equals $0$ \cite{AcinAndrianovJaneTarrach 2001}.  
The set of Type~4a states is the complement of the generic entanglement class, i.e., it is the set of  $3$-qubit wave functions that {\it cannot} be written in the form $A \ v_1\otimes v_2 \otimes v_3 \plus B \ v_4\otimes v_5 \otimes v_6$ with $\{v_1,\ldots,v_6\}$ linearly independent and both $A,B \neq 0$ \cite{AcinAndrianovJaneTarrach 2001}.

\begin{Theorem}
\label{CIS theorem for 3-qubit}
Let $\{f_1,g_1\},\{f_2,g_2\}$, and $\{f_3,g_3\}$ be three orthonormal bases of $\mathbb{C}^2$ and let 
\begin{equation}
\label{best Type 4a}
    \chi_* \equals A f_1\otimes f_2 \otimes f_3 \plus B_1 g_1\otimes f_2 \otimes f_3 \plus B_2 f_1\otimes g_2 \otimes f_3 \plus B_3 f_1\otimes f_2 \otimes g_3 
\end{equation}
be a Type~4a wave function in $\mathbb{C}^2 \otimes \mathbb{C}^2 \otimes \mathbb{C}^2$.  
If $\chi_*$ is a max-overlap Type~4a approximation of a $3$-qubit wave function $\psi$, then
\[
    \psi \equals  \langle \chi_*, \psi \rangle \ \chi_* \plus C\ g_1\otimes g_2 \otimes g_3
\]
for some coefficient $C$ such that $ |C|^2 = 1 - \big| \langle \chi_*, \psi \rangle \big|^2 $.
\end{Theorem}

\begin{proof}

One way to prove this is to use the hyperdeterminant invariant and mimic the proof of Lemma~\ref{Lemma for CIS 2}. 
Instead, we will use the a different technique, which could indeed also be used to prove Lemma~\ref{Lemma for CIS 2}.

Recall that the phase of the max-overlap approximation $\chi_*$ is such that $\langle  \chi_*, \psi \rangle > 0$.  

Let $\mathfrak{F} =  \Span\big\{ f_1\otimes f_2 \otimes f_3, \ g_1\otimes f_2 \otimes f_3, \ f_1\otimes g_2 \otimes f_3, \ f_1\otimes f_2 \otimes g_3 \big\}$.
Since $\chi_*$ is a max-overlap Type~4a approximation of $\psi$, it must be proportional to the projection of $\psi$ onto the subspace $\mathfrak{F}$, so that  
\begin{equation}
   \psi \equals 
     \langle \chi_*, \psi \rangle \ \chi_* \plus C \ g_1\otimes g_2 \otimes g_3   \plus D_1 \ f_1\otimes g_2 \otimes g_3 \plus  D_2 \ g_1\otimes f_2 \otimes g_3 \plus D_3\  g_1\otimes g_2 \otimes f_3 \ .
    \label{not expendible} 
\end{equation}  

Assuming first that $B_1B_2B_3 \ne 0$, we will show that the three coefficients $D_i$ all equal $0$. 
To do this we shall use stationarity conditions derived from the three rotations 
\[
   \begin{bmatrix} f_i(t) \\ g_i (t) \end{bmatrix}
    \equals 
   \begin{bmatrix}  \cos(t) & -e^{i\theta} \sin(t) \\     e^{-i\theta} \sin(t) & \ \cos(t)    \end{bmatrix} 
   \begin{bmatrix} f_i \\ g_i  \end{bmatrix} \ .
\]
The trajectory  
\[
 \chi(t) \equals 
      A \  f_1(t)\otimes  f_2 \otimes  f_3   \plus  B_1 \  g_1(t) \otimes  f_2 \otimes  f_3 
        \plus B_2 \ f_1(t) \otimes  g_2 \otimes  f_3     \plus B_3 f_1(t) \otimes  f_2 \otimes  g_3 
\]
lies within the set of Type~4a wave functions and equals $\chi_*$ when $t=0$.
Since $\chi_*$ is a maximizer of the functional $ \chi \longmapsto \big| \langle \chi,  \psi \rangle \big|$ on the set of Type~4a wave functions, and since $\langle  \chi_*, \psi \rangle > 0$, the real part $\mathrm{Re} \langle \chi(t), \psi \rangle $ must attain its maximum at $t = 0$, because 
\[
    \mathrm{Re} \langle \chi(t), \psi \rangle \ \le \ \big| \langle \chi(t), \psi \rangle \big| \ \le \ \big| \langle \chi_*, \psi \rangle \big|  
    \equals 
    \mathrm{Re} \langle \chi_*, \psi \rangle 
    \equals 
    \mathrm{Re} \langle \chi(0), \psi \rangle.
\]
The maximizer $\chi_*$ must therefore satisfy the following stationarity condition:
\begin{equation}
\label{stationarity condition}
  \frac{d}{dt} \mathrm{Re} \langle \chi(t),\psi \rangle \Big|_{t=0} \equals \mathrm{Re} \langle  \dot{\chi}(0), \psi \rangle \equals 0
\end{equation}
where $\dot{\chi}(0)$ denotes $\frac{d}{dt}  \chi(t) \big|_{t=0}$.
Since 
\[
 \dot{\chi}(0) \equals 
      e^{-i\theta}  B_1  f_1 \otimes  f_2 \otimes  f_3  \ - \ e^{i\theta}  ( A g_1\otimes  f_2 \otimes  f_3   \plus 
        B_2  g_1 \otimes  g_2 \otimes  f_3     \plus B_3 g_1 \otimes  f_2 \otimes  g_3 )\ ,
\]
the stationarity condition (\ref{stationarity condition}) implies that 
\[
     B_1 \langle \psi, f_1 \otimes  f_2 \otimes  f_3 \rangle   \equals    \overline{A} \langle g_1\otimes  f_2 \otimes  f_3, \psi \rangle 
     \plus \overline{B}_2  \langle g_1\otimes  g_2 \otimes  f_3, \psi  \rangle \plus \overline{B}_3  \langle g_1\otimes  f_2 \otimes  g_3, \psi  \rangle.
\]
Using (\ref{not expendible}) and (\ref{best Type 4a}) we find that  
$    B_1\langle  \psi, \chi_* \rangle \overline{A} \equals  \langle    \chi_*,\psi \rangle \big( \overline{A}  B_1 \plus \overline{B}_2  D_3 \plus \overline{B}_3 D_2 \big) $.
Since $\langle   \psi, \chi_* \rangle > 0$ is real, $\overline{B}_2  D_3 = -\overline{B}_3 D_2 $.
Similarly $\overline{B}_1  D_2 = -\overline{B}_2 D_1 $ and  $\overline{B}_3  D_1 = -\overline{B}_1 D_3 $.  
As we are assuming that $B_1B_2B_3 \ne 0$, the three preceding equations imply that $D_i = 0$, $i=1,2,3$. 

The degenerate cases, where one or more of the coefficients $B_i$ are $0$, can be handled exactly as in the proof of Theorem~\ref{Theorem for CIS}, thanks to the $3$-qubit correspondence.
\end{proof}

\acknowledgments{This work has been supported by the Austrian Science Foundation FWF,
project SFB F41 (VICOM) and project I830-N13 (LODIQUAS).}

\newpage

\noindent {\bf Appendix: another proof of the Borland-Dennis Theorem}

\medskip

Let $h$ be any natural orbital of $\psi$ and let $\MM$ denote the $5$-dimensional orthogonal complement of $\Span\{h\}$.  
Then $\psi$ can be written as $h \wedge \gamma + \phi$, where $\gamma \in \wedge^2 \MM$ and $\phi \in \wedge^3 \MM$.  
By Lemma~\ref{two-in-five} there exist orthonormal orbitals $f'_1,f'_2,g'_1,g'_2 \in \MM$ such that $\gamma = A_1 \ f'_1 \wedge f'_2 + A_2 \ g'_1 \wedge g'_2$.  

We shall prove the theorem separately in three cases, according to the rank of $\gamma$.  These cases are, namely, where (1) both $A_1$ and $A_2$ are nonzero, (2) only one of $A_1,A_2$ is nonzero, (3) both $A_1$ and $A_2$ equal $0$.

\noindent Case (1): 

Let $k$ be an orbital such that $(h,k,f'_1,f'_2,g'_1,g'_2)$ is an ordered orthonormal basis of $\Csix$.
The CI expansion of $\psi$ with respect to this basis is as shown in the following configuration diagram, 
whose terms have been arranged into four blocks for convenience:
\begin{equation*}
\begin{matrix}
 & h & k & f'_1 & f'_2 & g'_1 & g'_2 & \quad &  \\
                      & \mathrm{X} &  \mathrm{O} &  \mathrm{X} & \mathrm{X} & \mathrm{O} & \mathrm{O} & &  A_1   \\ 
  \mathrm{(i)}\  & \mathrm{X} &\mathrm{O} &  \mathrm{O} & \mathrm{O} & \mathrm{X} & \mathrm{X} & &  A_2  \\
 \hline
                      & \mathrm{O} &  \mathrm{X} & \mathrm{X} & \mathrm{X} & \mathrm{O} & \mathrm{O}  & &  B_1  \\
\mathrm{(ii)}\   & \mathrm{O} &  \mathrm{X} & \mathrm{O} & \mathrm{O} & \mathrm{X} & \mathrm{X}  & &  B_2     \\ 
\hline
                      & \mathrm{O}  &  \mathrm{X} & \mathrm{X} & \mathrm{O} & \mathrm{X} & \mathrm{O} & &   M_{11}  \\
                      & \mathrm{O}  &  \mathrm{X} & \mathrm{X} & \mathrm{O} & \mathrm{O} & \mathrm{X}  & &   M_{12}  \\ 
                      & \mathrm{O} &  \mathrm{X} & \mathrm{O} & \mathrm{X} & \mathrm{X} & \mathrm{O}  & &   M_{21}  \\ 
\mathrm{(iii)}\  & \mathrm{O} &  \mathrm{X}  & \mathrm{O} & \mathrm{X} & \mathrm{O} & \mathrm{X} & &   M_{22}  \\
\hline
                      & \mathrm{O} &  \mathrm{O} & \mathrm{X} & \mathrm{X} & \mathrm{X} & \mathrm{O} & & D_1    \\
                      & \mathrm{O} &  \mathrm{O} &  \mathrm{X} & \mathrm{X} & \mathrm{O} & \mathrm{X}  & & D_2  \\
                      & \mathrm{O} &  \mathrm{O} & \mathrm{X} & \mathrm{O} & \mathrm{X} & \mathrm{X}  & & D_3   \\
 \mathrm{(iv)}\  & \mathrm{O} &  \mathrm{O} & \mathrm{O} & \mathrm{X} & \mathrm{X} & \mathrm{X} & & D_4     \\ 
\end{matrix}
\end{equation*}

Let $\Gamma$ denote the 1RDM of $\psi$.  
Since $h$ is a natural orbital of $\psi$, the off-diagonal matrix elements of $\Gamma$ that connect $h$ to any other orbitals must all vanish.  
The only way this can happen is if $D_i = 0$, $i=1,\ldots,4$.  For example, the matrix element $\langle h, \Gamma f'_1 \rangle = \overline{A}_2 D_3$ must vanish; since $A_2 \neq 0$, it must be that $D_3 = 0$.  
Thus, all coefficients in block (iv) equal $0$.

The sum of terms in block (iii) is  
$   k \wedge\sum\limits_{i,j=1}^2 M_{ij} \  f'_i\wedge  g'_j $.
This can be reduced to a sum of two terms by using the same SVD technique that we used to complete the proof of Theorem~\ref{lone orbital theorem}.  The transformations (\ref{these substitutions}) result in the CI expansion   
\begin{equation}
\begin{matrix}
 h &  k & f_1 & f_2 & g_1 & g_2 & \quad &  \\
 \mathrm{X} &  \mathrm{O} &  \mathrm{X} & \mathrm{X} & \mathrm{O} & \mathrm{O} & &  A_1   \\ 
 \mathrm{X} &  \mathrm{O} &  \mathrm{O} & \mathrm{O} & \mathrm{X} & \mathrm{X} & &  A_2  \\
 \mathrm{O} & \mathrm{X} & \mathrm{X} & \mathrm{X} & \mathrm{O} & \mathrm{O} & &  B_1  \\
 \mathrm{O} & \mathrm{X} & \mathrm{O} & \mathrm{O} & \mathrm{X} & \mathrm{X} & &  B_2     \\ 
 \mathrm{O} & \mathrm{X} & \mathrm{X} & \mathrm{O} & \mathrm{X} & \mathrm{O}  & &   C_1  \\
 \mathrm{O} & \mathrm{X} & \mathrm{O} & \mathrm{X} & \mathrm{O} & \mathrm{X} & &   C_2  \\ 
\end{matrix}
\label{6-term}
\end{equation}
 for $\psi$ with respect to the ordered basis $(h,k,f_1,f_2,g_1,g_2)$.  
 This CI expansion shows that $\psi$ belongs to the $3$-qubit subspace $ \Span\{h,k\} \wedge  \Span\{f_1,g_2\} \wedge  \Span\{f_2,g_1\}$, and the conclusion of Theorem~\ref{Borland-Dennis Theorem} for this $\psi$ holds by Lemma~\ref{3-qubit subspace lemma}.

\noindent Case (2):

In this case, there exist orthonormal orbitals $f'_1,f'_2, \in \MM$, a trivector $\phi \in \wedge^3 \MM$, and a coefficient $A$ such that 
$\psi = A\ h \wedge  f'_1 \wedge f'_2 \plus \phi$.  
Let $g'_1,g'_2,g'_3$ be orbitals such that $\{f'_1,f'_2,g'_1,g'_2,g'_3\}$ is an orthonormal basis of $\MM$.
Let $\RR = \Span\{g'_1,g'_2,g'_3\}$ and let $\RR^\perp$ denote $\Span\{f'_1,f'_2\}$, the orthogonal complement of $\RR$ in $\MM$.    
Since 
\[
    \wedge^3\MM \equals \wedge^3\RR \oplus \big( \RR \wedge \RR \wedge \RR^\perp \big) \oplus \big( \RR \wedge \RR^\perp \wedge \RR^\perp \big)\ ,
\] the trivector $\phi \in \wedge^3 \MM$ can be written as a linear combination of the reference configuration $g'_1 \wedge g'_2 \wedge g'_3$ and configurations that are single and double excitations of the reference.  That is, 
$
    \phi = \varphi_{CIS} \plus \varphi_D
$
where $\varphi_{CIS} \in  \wedge^3\RR \oplus \big( \RR \wedge \RR \wedge \RR^\perp \big)$ 
and $\varphi_D \in \RR \wedge \RR^\perp \wedge \RR^\perp$.    
By Lemma~\ref{SVD lemma}, there exist orthonormal bases $\{g_1,g_2,g_3\}$ and $\{f_1,f_2\}$ of $\RR$ and $\RR^\perp$, respectively, and coefficients 
$B_1,B_2,C$ such that 
\[
       \varphi_{CIS} \equals B_1 f_1 \wedge g_2 \wedge g_3 \plus B_2 g_1 \wedge f_2 \wedge g_3 \plus C g_1 \wedge g_2 \wedge g_3
\]
Since $\RR^\perp  = \Span\{f_1,f_2\}$ is only $2$-dimensional, there are coefficients $D_1,D_2,D_3$ such that $\varphi_D = \sum\limits_{i=1}^3 D_i\ f_1 \wedge f_2 \wedge g_i$.   The CI expansion of $\psi$ with respect to $(h,f_1,f_2,g_1,g_2,g_3)$ is 
\begin{equation}
\begin{matrix}
 h &  f_1 & f_2 & g_1 & g_2 & g_3 \quad &  \\
 \mathrm{X} &  \mathrm{X} &  \mathrm{X} & \mathrm{O} & \mathrm{O} & \mathrm{O} & &  A   \\ 
 \mathrm{O} & \mathrm{X} & \mathrm{O} & \mathrm{O} & \mathrm{X} & \mathrm{X} & &  B_1  \\
 \mathrm{O} & \mathrm{O} & \mathrm{X} & \mathrm{X} & \mathrm{O} & \mathrm{X} & &  B_2     \\ 
 \mathrm{O} & \mathrm{O} & \mathrm{O} & \mathrm{X} & \mathrm{X} & \mathrm{X}  & &   C  \\
 \mathrm{O} & \mathrm{X} & \mathrm{X} & \mathrm{X} & \mathrm{O} & \mathrm{O} & &   D_1  \\ 
  \mathrm{O} & \mathrm{X} & \mathrm{X} & \mathrm{O} & \mathrm{X} & \mathrm{O} & &  D_2  \\ 
 \mathrm{O} & \mathrm{X} & \mathrm{X} & \mathrm{O} & \mathrm{O} & \mathrm{X} & &   D_3  \\ 
\end{matrix}
\label{7-term}
\end{equation}
Let $\Gamma$ denote the 1RDM of $\psi$.  
Since $h$ is a natural orbital of $\psi$,  the three matrix elements $\langle h, \Gamma g_i \rangle = \overline{A} D_i$ must equal $0$.
There remain only four terms in the CI expansion (\ref{7-term}) of $\psi$, examination of which shows that $\psi$ belongs to the $3$-qubit subspace 
$\Span\{h,g_3\} \wedge  \Span\{f_1,g_1\} \wedge  \Span\{f_2,g_2\}$.   The conclusion of Theorem~\ref{Borland-Dennis Theorem} holds by Lemma~\ref{3-qubit subspace lemma}.   

\noindent Case (3): 

Since $A_1$ and $A_2$ both equal $0$, $h$ has natural occupation number $0$ and $\psi  \in \wedge^3 \MM$ has rank $5$.
The canonical form of a rank-$5$ wave function given in Lemma~\ref{three-in-five} is itself a natural CI expansion of the form (\ref{eight}).


\begin{thebibliography}{99}
\bibitem{ChenDokovicGrasslZeng-footnote}
Proposition 4(ii) in Ref.~\cite{ChenDokovicGrasslZeng}.  The hypothesis stated there, that $M_5$ must also vanish, appears to be superfluous.
\end{thebibliography}
\end{document}